{\let~\catcode~13 13\def^^M{^^J}~` 12\xdef\asciiart{

   88b           d88                                  88    88                  
   888b         d888                                  88    ""                  
   88`8b       d8'88                                  88                        
   88 `8b     d8' 88     ,adPPYba,     88       88    88    88    8b,dPPYba,    
   88  `8b   d8'  88    a8"     "8a    88       88    88    88    88P'   `"8a   
   88   `8b d8'   88    8b       d8    88       88    88    88    88       88   
   88    `888'    88    "8a,   ,a8"    "8a,   ,a88    88    88    88       88   
   88     `8'     88     `"YbbdP"'      `"YbbdP'Y8    88    88    88       88

         db        88                        88                                 
        d88b       88                        88                                 
       d8'`8b      88                        88                                 
      d8'  `8b     88  ,adPPYb,d8  ,adPPYba, 88,dPPYba,  8b,dPPYba, ,adPPYYba,  
     d8YaaaaY8b    88 a8"    `Y88 a8P_____88 88P'    "8a 88P'   "Y8 ""     `Y8  
    d8""""""""8b   88 8b       88 8PP""""""" 88       d8 88         ,adPPPPP88  
   d8'        `8b  88 "8a,   ,d88 "8b,   ,aa 88b,   ,a8" 88         88,    ,88  
  d8'          `8b 88  `"YbbdP"Y8  `"Ybbd8"' 8Y"Ybbd8"'  88         `"8bbdP"Y8  
                       aa,    ,88                                               
                        "Y8bbdP"

}}\makeatletter\newif\iflabor%
\labortrue

\documentclass{amsart}

\usepackage{mathtools,amssymb}\allowdisplaybreaks

\usepackage[english]{babel}

\usepackage[utf8]{inputenc}
	\let\@DUC\DeclareUnicodeCharacter\@DUC{4EE4}{\@DO\@LET}
	\def\@DO#1{\bgroup\def\UTFviii@defined##1{\expandafter#1\string##1+}}
	\def\@LET#1:#2+{\egroup\@DUC{\UTFviii@hexnumber{\decode@UTFviii#2\relax}}}
	令定{\@DO\@DEF}\def\@DEF#1:#2+{\@LET:#2+{\@nameuse{ME#2AN}}\@namedef{ME#2AN}}
	
	令Γ\Gamma	令Δ\Delta	令Θ\Theta	令Λ\Lambda	令Ξ\Xi		令Π\Pi		
	令Σ\Sigma	令Υ\Upsilon	令Φ\Phi		令Ψ\Psi		令Ω\Omega	
	令α\alpha	令β\beta		令γ\gamma	令δ\delta	令ε\varepsilon			
	令ζ\zeta		令η\eta		令θ\theta	令ι\iota		令κ\kappa	令λ\lambda	
	令μ\mu		令ν\nu		令ξ\xi		令π\pi		令ρ\rho		令ς\varsigma	
	令σ\sigma	令τ\tau		令υ\upsilon	令φ\varphi	令χ\chi		令ψ\psi		
	令ω\omega	令ϑ\vartheta	令ϕ\phi		令ϖ\varpi	令Ϝ\Digamma	令ϝ\digamma	
	令κ\varkappa	令ϱ\varrho	令ϴ\varTheta	令ϵ\epsilon	
	令𝔽{\mathbb F}	
	令𝒜{\mathcal A}	令ℬ{\mathcal B}	令ℋ{\mathcal H}	令ℐ{\mathcal I}	
	令ℳ{\mathcal M}	
	\DeclareMathAlphabet\mathsi{T1}\sfdefault\mddefault\sldefault
	令𝘊{\mathsi C}		令𝘋{\mathsi D}		令𝘔{\mathsi M}		令𝘙{\mathsi R}	
	令𝘚{\mathsi S}		令𝘞{\mathsi W}	
	令ℓ\ell		令∂\partial	令∇\nabla	
	令√\sqrt		
	\def\bigol#1{\bigl#1\iflabor\else\color{UCO!50!black}\fi}
	\def\bigor#1{\iflabor\else\color{UIB!50!black}\fi\bigr#1}
	\def\({\bigol(}	\def\){\bigor)}	令（{\Bigl(}		令）{\Bigr)}		
	令［{\bigol[}	令］{\bigor]}	令「{\Bigl[}		令」{\Bigr]}		
	令｛{\bigol\{}	令｝{\bigor\}}	令『{\Bigl\{}	令』{\Bigr\}}	
	令⌈\lceil	令⌉\rceil	令⌊\lfloor	令⌋\rfloor	令⟨\langle	令⟩\rangle	
	\def\|{\mathrel\Vert}	令‖{\mathrel\Big\Vert}	令｜{\mid\nobreak}		
	令；{\mathrel;\nobreak}	令、{\setminus\nobreak}	令：{\colon}				
	令ˆ\hat		令˜\tilde	令¯\bar		令˘\breve	令˙\dot		令¨\ddot		
	令°\ocirc	令ˇ{^\vee}	
	令∏\prod		令∑\sum		令∫\int		令⋀\bigwedge	令⋁\bigvee	令⋂\bigcap	
	令⋃\bigcup	令⨁\bigoplus			令⨂\bigotimes			
	令±\pm		令·\cdot		令×\times	令÷\frac		令•\bullet	令∧\wedge	
	令∨\vee		令∩\cap		令∪\cup		令⊕\oplus	令⊗\otimes	令⋆\star
	令¬\neg		令∀\forall	令∃\exists	令∞\infty	令⊤\top		令⊥\bot		
	令⋯\cdots	令♠\spadesuit			令♡\heartsuit			
	令♢\diamondsuit			令♣\clubsuit	令♭\flat		令♮\natural	令♯\sharp	
	令←\gets		令→\to		令↞\twoheadleftarrow		令↠\twoheadrightarrow	
	令↤\mapsfrom	令↦\mapsto	令↩\hookleftarrow		令↪\hookrightarrow		
	令↾{\mathbin\upharpoonright}			
	令∈\in		令∉\notin	令∋\ni		令≅\cong		令≈\approx	令≔\coloneqq	
	令≠\neq		令≡\equiv	令≤\leqslant	令≥\geqslant	令⊆\subseteq	令⟂\perp		
	令⟵\longleftarrow		令⟶\longrightarrow		令⟼\longmapsto			
	令…{,\allowbreak\dotsc,\allowbreak}
	定†#1†{\text{#1}}
	定©#1©{{\color{UCO}#1}}
	定®#1®{{\color{Periwinkle}#1}}
	令🔑{^\star_}
	\def\[{\@ifstar{\begin{equation*}}{\begin{equation}}}
	\def\]{\@ifstar{\end  {equation*}}{\end  {equation}}}
	
	\DeclarePairedDelimiter\abs\lvert\rvert

	\DeclareMathOperator\spa{span}		
	\def\ce{_{\text{ce}}}				\def\co{_{\text{co}}}
	\def\V#1{T^{#1}V}	\def\Vp{\V p}	\def\W#1{Λ^{#1}W}	\def\Wq{\W q}
	\def\Ws{W^\star_{[k]}}
	
	\def\biy#1#2{\Bigl(\hbox{\smaller\!$\genfrac..\z@0{#1}{#2}$\!}\Bigr)}
	\def\biz#1#2{\Bigl(\hbox{\smaller[2]\!$\genfrac..\z@0{#1}{#2}$\!}\Bigr)}
	\let\bi=\biy
	\def\WM/{wedge-mul\-ti\-pli\-ca\-tion}
	\def\CM/{co\-wedge-mul\-ti\-pli\-ca\-tion}
	\def\coef/{co\-ef\-fi\-cient}
	\def\MBR/{M\kern-.2ex\lower.5ex\hbox{B}\kern-.2exR}
	\def\MSR/{M\kern-.2ex\lower.5ex\hbox{S}\kern-.2exR}
	\def\MDS/{M\kern-.2ex\lower.5ex\hbox{D}\kern-.2exS}
	定彈#1{\advance#1\glueexpr0ptplus1ptminus1pt}
	定邊#1{\marginpar{\color{2728 C}\flushleft #1}}
	令縮{\hskip-9cmplus9cm}

\usepackage{tikz-cd,pgfplotstable,booktabs,colortbl}
	\usetikzlibrary{calc,shapes.geometric}
			
	\let\PMT\pgfmathtruncatemacro	
	定色#1!#2 {\definecolor{#1}{HTML}{#2}}
	色UIB!13294b			色2728 C!0455A4		色2738 C!1F4096		
	色427!E8E9EA			色Cool Gray 6!A5A8AA	色Cool Gray 1!5E6669	
	色UCO!E84A27			色UIC red!D50032		色UIS blue!003366	
	色Teal!0d605e		色Gray-blue!6fafc7	色Citron!bfd46d		
	色Dark yellow!ffd125	色Salmon!ee5e5e		色Periwinkle!4f6898	
	定點(#1)[#2]{node(d#1)[circle,fill,inner sep=1]{}node(D#1)[anchor=#2]}
	定籤(#1)--+(#2:#3){(#1.#2)--+(#2:#3)node[anchor=#2+180]}
	定苯{node[regular polygon,regular polygon sides=6,inner sep=5,draw]}
	\tikzset{
		every picture/.style={cap=round,join=round},
		c/.tip={tikzcd right hook},
		thick dash/.style={dashed,line width=.8},
		shorten both/.style={shorten <=#1,shorten >=#1},
		sudoku/.style={xscale=1.5,yscale=3/4},
		subspace cell/.style={rounded corners=6,/to highlight/\x-\y/.try},
		know/.style={/to highlight/#1/.style={fill=Dark yellow}},
		learn/.style={/to highlight/#1/.style=draw},
	}
	\tikzcdset{
		skew 3/.style={row sep=1em,column sep=#1,
			cells={xshift=-\pgfmatrixcurrentrow*1em}},
		parity/.style={row sep=2em,column sep=1em},
		parity 2/.style={row sep=2em,column sep={between origins,3em}},
	}
	\def\CD{\catcode`\&=13\relax\CD@AUX}
	\newcommand\CD@AUX[2][]{\begin{tikzcd}[#1]#2\end{tikzcd}}
	\def\ARDS{\uar[xshift=.5em,<-c]\uar[xshift=-.5em,->>]
		\dar[xshift=.5em,->>]\dar[xshift=-.5em,<-c]}
	\newcommand\ARCM[1][]{\ar[start anchor={[xshift=-1em]south east},
		end anchor={[xshift=1.5em]north west}]{ddr}[pos=.25,#1]{∇}}
	\def\VW#1#2{\V{#1}⊗\W{#2}\ARCM}
	\def\VUW#1#2{\V{#1}⊗U⊗\W{#2}\ARDS}
	\def\VWW#1#2{\V{#1}⊗W⊗\W{#2}}
	\pgfplotsset{compat/show suggested version=false,compat=1.13} 
	\pgfmathsetseed{\day+31*(\month+12*(\year-2000))}

\usepackage[unicode,pdfusetitle,
	linkbordercolor=Salmon,citebordercolor=Teal]{hyperref}
	\def\PM#1$#2${\texorpdfstring{$#2$}{#1}}
	\def\PT#1†#2†{\texorpdfstring{#2}{#1}}
	\def\U#1+{\unichar{"#1}}

\usepackage[capitalize,noabbrev]{cleveref}
	
	定名#1:#2?#3?#4 {\crefname{#1}{#2#3}{#2#4}}		
	名page:page??s			名section:section??s		名enumi:item??s
	\newtheorem{thm}{Theorem}名thm:Theorem??s
	定理#1:#2?#3?#4 {\newtheorem{#1}[thm]{#2#3}名#1:#2?#3?#4 }令風\theoremstyle
	理cor:Corollar?y?ies		理lem:Lemma??s			理pro:Proposition??s		
	風{definition}			理dfn:Definition??s		理exa:Example??s			
	風{remark}				理cla:Claim??s			理rem:Remark??s			
	定式#1:#2?#3?#4 {名#1:#2?#3?#4 \creflabelformat{#1}{\textup{(##2##1##3)}}}
	式equ:equalit?y?ies		式con:containment??s		式dia:diagram??s			
	式for:formula??s			式fun:ogf??s				式ine:inequalit?y?ies	
	式spa:space??s			式ten:tensor??s			
	\def\@ReadTypeLabel#1:#2?{\xdef\@TYPE{#1}\xdef\@LABEL{#2}}
	\def\eqlabel#1{\@ReadTypeLabel#1?\label[\@TYPE]{\@TYPE:\@LABEL}}
	\def\steplabel{\incr@eqnum\tag\theequation\eqlabel}
	\newcommand\taglabel[2][0]{\@ReadTypeLabel#2?\label[\@TYPE]{\@TYPE:\@LABEL}
		\addtocounter{equation}{#1}\tag{\theequation.\@LABEL}}
	\def\bytag#1{\tag{by \eqref{#1}}}	\def\copytag#1{\tag{\eqref{#1}'s copy}}
	\def\itlabel#1{\hypertarget{item:#1}{(#1)}}
	\def\itref#1{\textup{(\hyperlink{item:#1}{#1})}}

定行{\catcode13 }行13定讀#1{行13\def^^M{\\}\def\@ATTRI{行5#1}\@ifnextchar[{選}{參}}
定選[^^M#1^^M]^^M#2^^M^^M{\@ATTRI[#1]{#2}}定參^^M#1^^M^^M{\@ATTRI{#1}}行5

讀
\title
				   Multilinear Algebra for Distributed Storage  				

讀
\author
				    Iwan Duursma\and Xiao Li\and Hsin-Po Wang   				

讀{\def
\@pdfsubject}
				         94B27; 15A75; math.IT; math.AC         				

讀
\subjclass
				         Primary 94B27; Secondary 15A75         				

讀
\thanks
		   This work was partially supported by NSF grant CCF-1619189.  		

\address{
				           Department of Mathematics,           				
				  University of Illinois at Urbana--Champaign,  				
				             Urbana, Illinois 61801             				
}
\email{%
				 duursma and xiaoli17 and hpwang2 @illinois.edu 				
}

\begin{document}\message{\asciiart}

\begin{abstract}
	An $(n, k, d, \alpha, \beta, M)$-ERRC (exact-repair regenerating code)
	is a collection of $n$ nodes used to store a file.
	For a file of total size $M$,
	each node stores $\alpha$ symbols, any $k$ nodes recover the file, and
	any $d$ nodes repair any other node via sending out $\beta$ symbols.
	We establish a multilinear algebra foundation to assemble
	$(n, k, d, \alpha, \beta, M)$-ERRCs for all meaningful $(n, k, d)$ tuples.
	Our ERRCs tie the $\alpha/M$-versus-$\beta/M$ trade-off with
	cascade codes, the best known construction for this trade-off.
	We give directions on how these ERRCs repair multiple failures.
\end{abstract}

彈\baselineskip/8 彈\lineskip/8 彈\parskip/8 彈\floatsep/4 彈\textfloatsep/4
\hbadness99\overfullrule1em

\maketitle

讀
\section
								  Introduction  								

\label{sec:intro}

	Distributed storage systems emerge as a nontraditional coding problem
	where the user gains and loses by multiples of
	a chunk of symbols called \emph{node}.
	The user wants to decode the original message
	by connecting to (only) a fraction of nodes.
	In addition, nodes themselves want to actively
	check for spontaneous erasures and refill them
	before the user asks a failing node for data.
	This motivates the following definition.
	
	\begin{dfn}\label{dfn:regenerate}
		\cite{DGWWR10,WD09,RSKR09}
		An \emph{$(n,k,d,α,β,M)$-ERRC} (exact-repair regenerating code)
		is a collection of $n$ nodes used to store an $M$-symbol file.
		The storage is configured such that
		(a)	each node stores $α$ symbols;
		(b)	any $k$ nodes contain sufficient information
			to recover the file; and
		(c)	any $d$ nodes repair any other failing node
			by sending out $β$ symbols.
	\end{dfn}
	
	In terms of random variables and entropies
	\cite[(4)--(6)]{Duursma14} \cite[Definition~1]{Tian14}:
	A file $Φ$ is a (random) vector in $𝔽^M$,
	where $𝔽$ is the working alphabet.
	Let nodes be indexed by integers $[n]≔\{1,2…n\}$.
	For each $h∈[n]$,
	the $h$th node stores a vector $𝘞_h∈𝔽^α$ depending on $Φ$.
	That is, $H(𝘞_h｜Φ)=0$ for all $h∈[n]$.
	The contents of any $k$ nodes recover the file $Φ$ in the manner that
	\[H(Φ｜𝘞_{h_1},𝘞_{h_2}…𝘞_{h_k})=0\eqlabel{equ:k}\]
	for arbitrary distinct indices $h_1,h_2…h_k∈[n]$.
	The actual procedure that recovers $Φ$ from $𝘞_{h_1},𝘞_{h_2}…𝘞_{h_k}$ is
	called the \emph{downloading scheme} or the \emph{data recovery scenario}.

	Nodes wear out.
	When the $f$th node fails for some $f∈[n]$,
	the physical being of the $f$th node is disconnected and discarded.
	An empty, brand new replacement will be plugged into the system;
	it is called a \emph{newcomer}.
	A subset $ℋ⊆[n]、\{f\}$ of $d$ nodes will be asked to help
	reproduce the data on the newcomer.
	To that end, a helper node of index $h∈ℋ$
	sends out a helping vector $𝘚^ℋ_{h→f}∈𝔽^β$.
	In formal language, $H(𝘚^ℋ_{h→f}｜𝘞_h)=0$
	for all $f∈[n]$ and all $h∈ℋ⊆[n]、\{f\}$.
	The help messages contain sufficient information
	to repair the $f$th node in the manner that
	\[H(𝘞_f｜𝘚^ℋ_{h_1→f},𝘚^ℋ_{h_2→f}…𝘚^ℋ_{h_d→f})=0\eqlabel{equ:d}\]
	for arbitrary distinct indices $f,h_1,h_2…h_d∈[n]$ and $ℋ≔\{h_1,h_2…h_d\}$.
	The actual procedure that recovers $𝘞_f$
	from $𝘚^ℋ_{h_1→f},𝘚^ℋ_{h_2→f}…𝘚^ℋ_{h_d→f}$ is called
	the \emph{repairing scheme} or the \emph{node repairing scenario}.
	
	This definition immediately poses a dilemma.
	In order to store files more efficiently,
	node contents should share very little mutual information.
	But then, repairing worn-out nodes becomes more difficult as it is hard
	to find relations among vectors sharing little mutual information.
	The parameter $β$ is referred to as the \emph{repair bandwidth}
	as it represents the required bandwidth of the network
	(from a helper to the newcomer) when the failing node
	needs to be reconstructed within a time limit.
	Another interpretation is that $dβ/α$ is the average length of
	the parity check equations used to reconstruct symbols in the $f$th node.
	
	\cite{DGWR07b,DGWR07n} proposed the earliest prototype
	of \cref{dfn:regenerate} and addressed the dilemma.
	They used network coding techniques to pinpoint
	the minimally required $β$ when $M=kα$---in other words,
	when \cref{equ:k} is achieved without overhead.
	This regime is later referred to as the \emph{\MSR/ point}.
	They also pinpoint the minimally required $α$ when $α=dβ$%
	---when \cref{equ:d} is achieved without overhead.
	And this is called the \emph{\MBR/ point}.
	
	Later in \cite{DGWWR10},
	a family of trade-offs between $α$ and $β$ was posed.
	They first defined what is now called \emph{FRRC}
	(functional-repair regenerating code)
	that is more general than \cref{dfn:regenerate}.
	In an FRRC, the newcomer node does not store
	the exact same content as the failing being used to.
	Instead, the newcomer will store whatever that is appropriate
	for the system  to sustain
	(to maintain its functionality as a device that stores the file $Φ$).
	The trade-off is best explained by \cref{fig:433}:
	Fix, say, $(n,k,d)=(4,3,3)$.
	Every triple $(α,β,M)$ such that the homogeneous pair $(α/M,β/M)$ lies
	strictly below the solid segments violates some information-theoretic
	inequalities, and hence is unfeasible. 
	Similar segments are identified for all $(n,k,d)$ triples.
	On the other hand, every (rational) point (inclusively) above
	the solid segments is achievable by some $(n,k,d,α,β,M)$-FRRC
	using the network coding techniques.
	
	Following the initial result, \cite{Wu10} gave an lower bound
	on the size of the field over which the codes can be implemented.
	They also relaxed the restriction that the FRRCs found in \cite{DGWWR10}
	only survive a prescribed number of failing--repairing rounds.
	Put in another way, \cite{Wu10}'s FRRCs can survive
	an unbounded number of failures,
	as long as every failure takes place after the previous failure is fixed.
	
	\begin{figure}
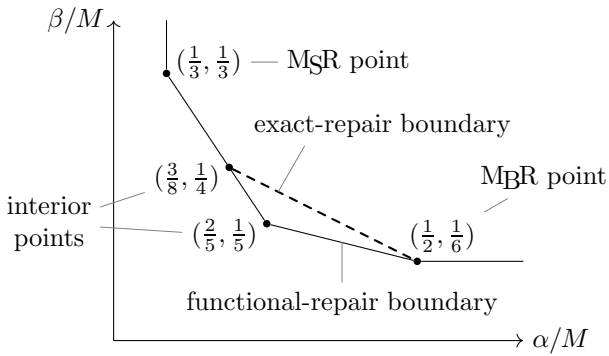
%
	\labortrue
		$$\tikz[x=20cm,y=15cm]{
			\def\ymin{1/6}\def\ymax{1/3}\def\xmin{1/3}\def\xmax{1/2}
			\draw[<->]
				($(\xmin,\ymax)+(-2em,2em)$)coordinate(Y)node[left]{$β/M$}|-
				($(\xmax,\ymin)+(4em,-3em)$)coordinate(X)node[right]{$α/M$};
			\iflabor
			\draw
				(1/3,1/3)點(1)[195]{$(÷13,÷13)$}
				(3/8,1/4)點(2)[15]{$(÷38,÷14)$}
				(2/5,1/5)點(3)[15]{$(÷25,÷15)$}
				(1/2,1/6)點(4)[225]{$(÷12,÷16)$}
				(Y-|d1)--(d1)--(d3)--coordinate[pos=.5](e34)(d4)--(d4-|X);
			\draw[thick dash](d2)--coordinate[pos=.25](e3)(d4);
			\draw[help lines,nodes=black]
				籤(D1)--+(0:1em){\MSR/ point}
				籤(e3)--+(60:2em)[anchor=180+15]{exact-repair boundary}
				籤(D4)--+(30:1em){\MBR/ point}
				籤(e34)--+(-90:1.5em){functional-repair boundary}
				(Y|-d3)+(-.5em,0)node(IP)[left,align=center]{interior\\points}
				(D2)--(IP)--(D3);
			\fi
		}$$
		\caption{
			The $α/M$-versus-$β/M$ trade-off for $(k,d)=(3,3)$.
			Point $(2/5,1/5)$ is achieved by some FRRC
			(functional-repair regenerating code).
			The other marked points are achieved by some ERRCs
			(exact-repair regenerating code).
			(Note that neither axis starts from~$0$.)
		}\label{fig:433}
	\end{figure}
	
	The notion of ERRC (especially the exact repair part)
	dated back to \cite{WD09,RSKR09} who focused on \MSR/ and \MBR/ points.
	Whether or not ERRCs achieve the functional repair segments was unsettled
	until \cite{Tian14} came with a negative answer.
	They first turned this into a linear programming problem
	and let computers solve it.
	The output implies that, as in \cref{fig:433},
	points below the dashed segments are unfeasible.
	Their proof is dedicated to the $(n,k,d)=(4,3,3)$ case but it was
	widely believed that this phenomenon persists for general parameters.
	Confirmed by \cite{SSK14} is that there is always a gap
	between ERRCs and the functional repairing boundary.
	Although it is unclear how large the gap exactly is.
	
	Afterward, more bounds on the infeasibility side are found and refined.
	To name a few, \cite{Duursma14,
		PK15, 
		Tian15,SPKVSK16, 
		EMT15, 
		MT15, 
		LL16, 
		Duursma19}. 
	Meanwhile, there are works devoted to constructing new ERRCs
	to approach the infeasibility bound from the other side.
	See \cite{
		RSK11, 
		SRKR12d,SRKR12i, 
		TSAVK15, 
		SSK15, 
		GEC14, 
		EM16d,EM19c,
		DL19}. 
	There are other works that concentrate exclusively at
	the \MSR/ point and pursue additional properties
	such as optimal access and low sub-packetization level.
	See, for example, \cite{
	GLJ18, 
	VRPKLSKBYNHN18, 
	YB19, 
	CYB20}. 
	
	Up until now, Elyasi--Mohajer's \emph{cascade codes}
	in \cite{EM19c} have achieved the best-known $(α/M,β/M)$-pairs
	across all meaningful parameters (i.e., $n-1≥d≥k≥1$).
	This work is based on their earlier work on determinant codes \cite{EM16d}:
	One concatenates several copies of determinant codes
	of various parameters such that some overprotected fragments
	yield their redundancies to another insecure fragment.
	This parents-protecting-child relationship is nested
	such that an insecure fragment might become
	overprotected after receiving redundancies;
	it then has to yield the extra redundancies to its successors.
	Playing with the family tree, Elyasi--Mohajer came up with a general rule
	of how to redistribute redundancies and cascade coding is born.
	
	Our works on this topic went parallel to Elyasi--Mohajer.
	We first made a connection between layered code and determinant code
	in terms of multilinear algebra in a conference talk \cite{LD17}.
	After cascade coding went out, we generalized
	our algebra-aided ERRC to mimic the concatenating nature
	of cascade using what we called \emph{\CM/}.
	(And this work ended up unpublished.)
	At the same time, we were inspired by the layered coding \cite{TSAVK15}
	and the improved layered coding \cite{SSK15} and proposed the
	Johnson graph codes \cite{DL19} with a combinatorial flavor.
	Lastly, we invent a completely new family of ERRCs
	that does not look like anything above and present it in this paper.
	Further connections are pointed out in \cref{sec:connection}.
	
	The ERRCs to be presented in this work attain
	the same set of parameters as cascade codes do.
	With the tensor and wedge notations we are able to
	largely reduce exceptional treatments of edge cases,
	endless bookkeeping of indices, and the lexicographical ordering.
	With generating functions we lessen the complexity of enumeration problems.
	
	We now state the main theorem.
	
	\begin{thm}[main theorem]\label{thm:moulin}
		For any integers $n$, $k$, $d$, and $s$ such that $n-1≥d≥k≥s-1≥1$,
		there exists an $(n,k,d,α,β,M)$-ERRC with parameters
		\begin{align}
			α &= ∑_{p+q=s-1}(d-k)^p\bi kq,\eqlabel{for:alpha}\\
			β &= ∑_{p+q=s-2}(d-k)^p\bi{k-1}q,
				\rlap{\qquad and}\eqlabel{for:beta}\\
			M &= ∑_{p+q=s-1}d(d-k)^p\bi kq-∑_{p+q=s}(d-k)^p\bi kq
				\eqlabel{for:em}
		\end{align}
		over every field of size $n$ or greater.
		Here $p,q≥0$.
	\end{thm}
	
	We name it \emph{moulin code} inspired by
	cascade (waterfall) and \emph{mu}lti\emph{lin}ear algebra.
	
	The parameter $s$ in the theorem is an auxiliary parameter
	that reflects the \emph{size} or \emph{scale} of a code.
	The counterparts are the \emph{mode} in \cite{EM16d,EM19c}
	and the \emph{layer size} in \cite{TSAVK15,SSK15,DL19}.
	An $s=2$ code is always at the \MBR/ point;
	an $s=k+1$ code is always at the \MSR/ point.
	A larger size/scale means the code is mentally more intricate
	although $α$, $β$, or $M$ is not necessarily greater.
	There is a subjectively better way to express the parameters using
	ordinary generating functions, which is not fully exploited by \cite{EM19c}.
	
	\begin{pro}[compact parameter]\label{pro:ogf}
		In \cref{thm:moulin}, $α$, $β$, and $M$ are the $x^s$-\coef/s of
		\begin{align}
			𝒜_{k,d}(x) &≔ ÷{x(1+x)^k}{1-(d-k)x},\eqlabel{fun:alpha}\\
			ℬ_{k,d}(x) &≔ ÷{x^2(1+x)^{k-1}}{1-(d-k)x},
				\rlap{\qquad and}\eqlabel{fun:beta}\\
			ℳ_{k,d}(x) &≔ ÷{(-1+dx)(1+x)^k}{1-(d-k)x},\eqlabel{fun:em}
		\end{align}
		respectively.
	\end{pro}
	
	Before the proofs of \cref{thm:moulin,pro:ogf},
	we will give an $s=4$ example in \cref{sec:example}.
	The formal, general construction will be given in \cref{sec:general}.
	Within the \cref{sec:example,sec:general}, one shall see that
	\CM/ is pivotal to the parity checks that define the code.
	The downloading scheme then uses parity checks to propagate
	the knowledge (belief) of symbols (variables) to the entire domain,
	which helps us comprehend the file after accessing $k$ nodes.
	In addition to \CM/, our repairing scheme relies heavily on
	the \emph{coboundary operators} (to be defined in \cref{sec:tailor}).
	They get the name after algebraic topology and have vanishing squares.
	Some of their interesting properties
	are derived in \cref{sec:tailor,sec:example}.
	It is worth noting that
	\cref{lem:homotopy} has the form of a cochain homotopy.
	It is also noteworthy that \cref{lem:bandwidth}
	is essentially dealing with the dimensions of the cohomology groups.
	
	We also prepare for catastrophic scenarios:
	The next proposition reveals the cost of repairing multiple failures.
	Previously, \cite{EM19d} did a similar analysis
	that is limited to the $k=d$ case.
	
	\begin{pro}[bulk repair]\label{pro:more}
		In emergency, ERRCs constructed for \cref{thm:moulin}
		repair $c≥1$ failing nodes at once.
		In repairing, every one of $d$ helper nodes sends out $β_c$ symbols,
		where
		\[β_c=∑_{p+q=s-2}(d-k)^p（\bi k{q+1}-\bi{k-c}{q+1}）,
			\eqlabel{for:betamore}\]
		and is the $x^s$-\coef/ of
		\[ℬ_{k,d,c}(x)≔𝒜_{k,d}(x)（1-÷1{(1+x)^c}）.\eqlabel{fun:betamore}\]
		In particular, $ℬ_{k,d}(x)=ℬ_{k,d,1}(x)=𝒜_{k,d}(x)x/(1+x)$.
	\end{pro}
	
	\Cref{sec:multiple} has the proof.

\subsection{Organization}

	\Cref{sec:back} reviews some algebra background,
		especially tensor and exterior algebras.
	\Cref{sec:tailor} defines U-, V-, and W-spaces, 
		\CM/, and coboundary operators.
		It also proves handy lemmas.
	\Cref{sec:example} demonstrates moulin design with an $s=4$ example.
	\Cref{sec:general} declares general moulin construction, and verifies that
		it produces $(n,k,d,α,β,M)$-ERRCs as described in \cref{thm:moulin}.
	\Cref{sec:multiple} analyzes how moulin survives multiple node failures.

讀
\section
							    Algebra Backbone    							

\label{sec:back}

	This section serves as a self-contained introduction
	to tensor and wedge algebras that will be used in the code construction.
	Contents of this section can be found in standard textbooks.
	To skip, proceed to \cref{sec:tailor} on \cpageref{sec:tailor}.
	
	Let $𝔽$ be a field.
	We measure information in $𝔽$-symbols
	so the finiteness of $𝔽$ is not mandatory.
	However, finite fields---especially those with characteristic $2$---%
	are assumed for applications (distributed storage as in the title).
	On the other hand, a crucial part of the construction
	implies that the field must have sufficiently many elements;
	we elaborate the implication later in \cref{sec:field}.
	
	Let $U$, $V$, and $W$ be finite dimensional vector spaces over $𝔽$.
	Elements of $U$ are denoted by $u$ (with or without subscripts),
	elements of $V$ by $v$, and elements of $W$ by $w$.
	For brevity, we call vector spaces \emph{spaces}.
	
	Denoted by $Uˇ$, the \emph{dual space} of $U$ is the space
	consisting of all linear transformations from $U$ to $𝔽$.
	We call elements of $Uˇ$ \emph{functionals} to distinguish them
	from elements of $U$, which we call \emph{vectors}.
	Since $U$ is of finite dimension, $U$ and $Uˇ$ share the same dimension.
	Furthermore, $(Uˇ)ˇ$ is isomorphic to $U$ canonically---a vector $u∈U$ gives
	rise to a map from $Uˇ$ to $𝔽$ by mapping a functional $ϕ∈Uˇ$ to $ϕ(u)∈𝔽$.
	It turns out that linear transformations defined in this way
	exhaust all possible linear transformations from $Uˇ$ to $𝔽$.
	The field element $ϕ(u)∈𝔽$ is called the \emph{evaluation of $ϕ$ at~$u$}.
	The action that takes a functional $ϕ∈Uˇ$ as input and returns $ϕ(u)∈𝔽$
	is called \emph{evaluating $ϕ$ at $u$}.
	When $ϕ$ is understood from the context, we simply say
	the \emph{evaluation at $u$} and \emph{evaluating at $u$}.
	For any subspace $V⊆U$, the \emph{restriction of $ϕ$ to $V$}
	is a functional from $V$ to $𝔽$ that evaluates $v∈V⊆U$ to $ϕ(v)$.
	This restriction is denoted by $ϕ↾V$.
	The corresponding action is called \emph{restricting $ϕ$ to $V$}.
	When $ϕ$ is understood from the context, we simply say
	the \emph{restriction to $V$} and \emph{restricting to $V$}.
	
	A crucial part of our construction involves evaluations of
	a functional $ϕ∈Uˇ$ at a list of vectors $u_1,u_2,u_3,\dotsc∈U$.
	Interesting things happen when these vectors share linear relation.
	For instance, if we want to evaluate $ϕ∈Uˇ$ at $u_1$, $u_2$, and $u_1-3u_2$,
	then we can also evaluate at the first two vectors ($u_1$ and $u_2$) and
	compute the third evaluation by linearity $ϕ(u_1-3u_2)=ϕ(u_1)-3ϕ(u_2)$.
	The information content of $ϕ(u_1)$, $ϕ(u_2)$, and $ϕ(u_1-3u_2)$
	is no more than that of $ϕ(u_1)$ and $ϕ(u_2)$.
	More generally, if $V$ is a subspace of $U$ and we want to know
	the restriction $ϕ↾V$, it suffices to choose a basis of $V$
	(any basis) and evaluate at each vector in the basis.
	For all intents and purposes,
	which basis is used does not affect the properties of the codes;
	only the cardinality of the basis, $\dim(V)$, matters.
	
	Let $U⊕V$ be the \emph{direct sum} of two spaces $U$ and $V$.
	This space consists of elements of the form $(u,v)$ where $u∈U$ and $v∈V$.
	The addition is defined as $(u,v)+(u',v')≔(u+u',v+v')$
	for all $u,u'∈U$ and all $v,v'∈V$.
	The scalar multiplication is defined as $c·(u,v)≔(cu,cv)$ for any $c∈𝔽$.
	The dimension is $\dim(U⊕V)=\dim(U)+\dim(V)$.
	It is possible to define the direct sum of three spaces $U,V,W$
	by $(U⊕V)⊕W$ or $U⊕(V⊕W)$ or in any other order of preference.
	These possibilities are \emph{not} a priori the same space
	but they are all isomorphic to each other.
	It is common to unify $(U⊕V)⊕W$ and $U⊕(V⊕W)$ as $U⊕V⊕W$ and
	treat it as a space consisting of elements of the form $(u,v,w)$.
	The addition is coordinate-wise;
	the scalar multiplication is distributive.
	For the direct sum of four or more spaces, the same guideline rules.
	
	Another ``where to put parenthesis'' problem arises
	when we want to combine dual space and direct sum.
	The space $Uˇ⊕Vˇ$ is isomorphic
	to $(U⊕V)ˇ$ in a straightforward manner.
	Similar isomorphisms exist for cases with three or more factors;
	for instance $Uˇ⊕V⊕Wˇ$ is isomorphic to $(U⊕Vˇ⊕W)ˇ$.
	
	Let $ϕ∈(U⊕V⊕W)ˇ$ be a functional.
	Normally we evaluate $ϕ$ at a triple $(u,v,w)∈U⊕V⊕W$.
	When $v$ and $w$ are zero vectors, we call $ϕ(u,0,0)$
	the \emph{evaluation of $ϕ$ at $u$}.
	Put in another way, the evaluation at $u$ is done via
	treating $u$ as its canonical copy $(u,0,0)$ in $U⊕V⊕W$.
	Similarly, we can evaluate $ϕ$ at $v∈V$ by evaluating at $(0,v,0)$,
	and at $w∈W$ by evaluating at $(0,0,w)$.
	We call $ϕ↾U$ the \emph{restriction (of $ϕ$) to $U$}, which, in actuality,
	is a functional from $U$ to $𝔽$ that evaluates $u$ to $ϕ(u,0,0).$
	Restrictions $ϕ↾V$ and $ϕ↾W$ are defined likewise.
	
	Since any compounded vector $(u,v,w)$ is a sum $(u,0,0)+(0,v,0)+(0,0,w)$
	and $ϕ$ is linear, the evaluation of $ϕ$ at $(u,v,w)$
	is the sum of evaluations at $(u,0,0)$, at $(0,v,0)$ and at $(0,0,w)$.
	Colloquially, evaluations of $ϕ$ are determined by
	the evaluations at $U$, those at $V$, and those at $W$.
	More concisely, $ϕ$ is determined by $ϕ↾U$, $ϕ↾V$, and $ϕ↾W$.
	
	Direct sum of spaces generalizes to
	direct sum of mappings in the following regard.
	Let $ϕ：U→W$ be a linear transformation,
	then there is a linear transformation $Φ：V⊕U→V⊕W$
	that sends $(v,u)∈V⊕U$ to $(v,ϕ(u))∈V⊕W$.
	In other words, $Φ$ applies $ϕ$ to the designated slot,
	and leaves the other slot intact.

\subsection{Tensors and tensor products}

	Let $U$ have dimension $d$ and a basis $\{¯u_1,¯u_2…¯u_d\}$.
	Let $V$ have dimension $l$ and a basis $\{¯v_1,¯v_2…¯v_l∈V\}$.
	The \emph{tensor product of $U$ and $V$}, denoted by $U⊗V$,
	is the space that consists of formal sums of the form
	\[∑_{ij}a_{ij}¯u_i⊗¯v_j.\eqlabel{for:tensorbasis}\]
	Here $a_{ij}∈𝔽$, and each $¯u_i⊗¯v_j$ is an unbreakable,
	free variable whose sole purpose is to carry its coefficient.
	The addition is term-wise.
	\[*∑_{ij}a_{ij}¯u_i⊗¯v_j+∑_{ij}b_{ij}¯u_i⊗¯v_j
		≔∑_{ij}(a_{ij}+b_{ij})¯u_i⊗¯v_j.\]*
	The scalar multiplication is distributive.
	\[*c·∑_{ij}a_{ij}¯u_i⊗¯v_j≔∑_{ij}(ca_{ij})¯u_i⊗¯v_j.\]*
	The dimension is $\dim(U⊗V)=\dim(U)·\dim(V)=dl$.
	
	We could have put $a_{ij}$ into a $d$-by-$l$ array
	and define $U⊗V$ to be the space of arrays (matrices).
	However, doing so prevents us from seeing the greater picture:
	We can turn the character ``$⊗$'' into
	an infixed binary operator from $U⊕V$ to $U⊗V$ that sends
	\[*(u,v)=（∑_ia_i¯u_i,∑_jb_j¯v_j）∈U⊕V,\]*
	where $a_i,b_j∈𝔽$, to
	\[u⊗v≔∑_{ij}(a_ib_j)¯u_i⊗¯v_j∈U⊗V.\eqlabel{for:tensorany}\]
	This map is \emph{bi-linear} in the sense that it is linear in $u$, meaning
	\[*(u+cu')⊗v=∑_{ij}(a_ib_j+ca_i'b_j)¯u_i⊗¯v_j=u⊗v+c(u'⊗v),\]*
	and linear in $v$, meaning
	\[*u⊗(v+cv')=∑_{ij}(a_ib_j+ca_ib_j')¯u_i⊗¯v_j=u⊗v+c(u⊗v'),\]*
	but not in both,
	meaning that $(u+cu')⊗(v+cv')$ is generally not $u⊗v+cu'⊗v'$.
	(The expansion should be $u⊗v+cu⊗v'+cu'⊗v+c^2u'⊗v'$.)
	Once we give $u⊗v$---the juxtaposition of ``$⊗$'' with arbitrary vectors---%
	an interpretation, we describe an element of $U⊗V$ through summing
	a finite list of $u_i⊗v_i$, where $u_i∈U$ and $v_i∈V$ are arbitrary vectors.
	We then treat $U⊗V$ as the collection of
	sums of the form $∑_ia_iu_i⊗v_i$.
	The addition is done via adding the coefficients
	of the matched $(u_i⊗v_i)$-terms and leaving unmatched terms intact.
	For example, $(2u_1⊗v_1+u_2⊗7v_2)$ plus $(-u_2⊗v_2+u_3⊗8v_3)$
	is equal to $(2u_1⊗v_1+6u_2⊗v_2+8u_3⊗v_3)$.
	In this new syntax, $u_i$ (or $v_i$) might not form the same basis
	of $U$ (or of $V$) as those $¯u_i,¯v_j$ in \cref{for:tensorbasis} do;
	they might not form a basis at all.
	A corollary is that, no matter which particular basis
	we choose in \cref{for:tensorbasis}, we will end up defining
	the one vector space structure on $U⊗V$, up to isomorphism.
	
	We call an element of $U⊗V$ a \emph{tensor} to distinguish it
	from \emph{vectors}, which are elements of plainer spaces like $U,V,W$.
	The fact that $-u_1⊗v_1-u_2⊗v_2+u_1⊗v_2+u_2⊗v_1$ and
	$u_2⊗(-v_2+v_1)-u_1⊗(v_1-v_2)$ along with $(-u_1+u_2)⊗v_1+(u_1-u_2)⊗v_2$ as
	well as $(u_2-u_1)⊗(v_1-v_2)$ describe the same tensor inspires a question,
	What is the least amount of ``$⊗$'' required to describe a tensor?
	In a tensor product of two spaces,
	this question boils down to decomposing a matrix $[a_{ij}]_{ij}$ into
	a product $𝘊𝘙$ of a $d$-by-$r$ matrix $𝘊$ and an $r$-by-$l$ matrix $𝘙$
	with the least possible $r$.
	(Remark: when $r$ reaches the minimum,
		columns of $𝘊$ are a basis of the column space of $[a_{ij}]_{ij}$;
		rows of $𝘙$ are a basis of the row space.)
	The number $r$ is called the \emph{rank of a tensor},
	which resembles the rank of a matrix.
	When $r=1$, the tensor is of the form $au⊗v$ for $a∈𝔽$ and $(u,v)∈U⊕V$.
	This is called a \emph{rank-$1$ tensor} or a \emph{simple tensor}.
	
	The new tensor notation defined in \cref{for:tensorany}
	possesses more convenience than \cref{for:tensorbasis} does.
	Consider again the tensor product $U⊗V$.
	We interpret $u⊗V$ as the collection of tensors of the form $∑a_iu⊗v_i$,
	that is, the sums where the ``$U$-component'' is always $u$.
	We interpret $U⊗v$ as the collection of tensors of the form $∑a_iu_i⊗v$.
	If $W$ is a subspace of $U$, then we interpret $W⊗V$ as the collection
	of tensors where the ``$U$-component'' is always in $W$.
	Clearly $u⊗V$, $U⊗v$, and $W⊗V$ are subspaces of $U⊗V$.
	
	The tensor notation generalizes to combinations of three or more spaces.
	Let $U$ and $V$ have bases $\{¯u_1,¯u_2…¯u_d\}$
	and $\{¯v_1,¯v_2…¯v_l\}$, respectively.
	Let $W$ be a $k$-dimensional space with a basis $\{¯w_1,¯w_2…¯w_k\}$.
	Not surprising is that $U⊗(V⊗W)$, $(U⊗V)⊗W$,
	and any other combination give the same vector space structure.
	It is common to unify them as $U⊗V⊗W$, a space consisting of
	formal sums of the form $∑_{hij}a_{hij}¯u_h⊗¯v_i⊗¯w_j$.
	The addition is term-wise.
	The scalar multiplication is distributive.
	The dimension is $\dim(U⊗V⊗W)=\dim(U)·\dim(V)·\dim(W)=dlk$.
	Similar to \cref{for:tensorany}, we interpret
	\[*u⊗v⊗w=（∑_ha_h¯u_h）⊗（∑_ib_i¯v_i）⊗（∑_jc_j¯w_j）,\]*
	where $(u,v,w)∈U⊕V⊕W$ and $a_h,b_i,c_j∈𝔽$, as
	\[*∑_{hij}(a_hb_ic_j)¯u_h⊗¯v_i⊗¯w_j∈U⊗V⊗W.\]*
	This defines a ternary operator $•⊗•⊗•$ that is \emph{tri}-linear
	in the sense that $(u+cu')⊗v⊗w=u⊗v⊗w+c(u'⊗v⊗w)$ and
	$u⊗(v+cv')⊗w=u⊗v⊗w+c(u⊗v'⊗w)$ along with $u⊗v⊗(w+cw')=u⊗v⊗w+c(u⊗v⊗w')$.
	This evidently provides a versatile way to describe tensors in $U⊗V⊗W$.
	Namely, a sum $∑_ia_iu_i⊗v_i⊗w_i$ is a tensor.
	We ask again what the least possible length of sums
	that describe a certain tensor is, and call this number its \emph{rank}.
	And then we can talk about whether a tensor is of rank one or not;
	a rank-$1$ tensor is of the form $au⊗v⊗w$ for $a∈𝔽$ and $(u,v,w)∈U⊕V⊕W$.
	Every tensor is a sum of several rank-$1$ tensors.
	To rephrase it, rank-$1$ tensors span a tensor product.
	A critical consequence is that we can describe a linear transformation
	from a tensor product by describing the image of every rank-$1$ tensor.
	
	The dual of a tensor product is the tensor product of duals,
	e.g., $(U⊗Vˇ⊗W)ˇ$ is isomorphic to $Uˇ⊗V⊗Wˇ$.
	Let $ϕ∈(U⊗V⊗W)ˇ$ be a functional.
	(We do not have a word to distinguish plain functionals in $Uˇ$, $Vˇ$,
	and $Wˇ$ from tensor-flavored functionals in $(U⊗V⊗W)ˇ$ and the like.)
	Since every tensor is a sum of rank-$1$ tensors, describing $ϕ$
	is equivalent to describing $ϕ$'s evaluations at rank-$1$ tensors.
	Even more generally, when $ϕ∈(U⊗U⊕V⊗W⊗V)ˇ$ is a very complicated functional,
	it is in fact determined by the restrictions to the direct summands,
	$ϕ↾U⊗U$ and $ϕ↾V⊗W⊗V$.
	For $ϕ↾U⊗U$, it reduces to understanding evaluations
	at rank-$1$ tensors of the form $u_1⊗u_2$;
	for $ϕ↾V⊗W⊗V$, it reduces to understanding evaluations
	at rank-$1$ tensors of the form $v_3⊗w_4⊗v_5$.
	
	Tensor products of spaces generalize to
	tensor products of mappings in the following regard.
	Let $ϕ：U→W$ be a linear transformation,
	then there is a linear transformation $Φ：V⊗U→V⊗W$
	that sends rank-$1$ tensors $v⊗u∈V⊗U$ to $v⊗ϕ(u)∈V⊗W$.
	In other words,
	we can prepend the mapping process $u↦ϕ(u)$ by the prefix ``$v⊗{}$''.
	Expression of a high-rank tensor as a sum of rank-$1$ tensors
	is not unique, but the result of applying $Φ$ would not change.

\subsection{Tensor powers and exterior powers}

	Let $\V0$ be $𝔽$;
	let $\V1$ be $V$;
	and let $\Vp$ be a product $V⊗V⊗\dotsb⊗V$ of $p$ many $V$'s.
	This is called the \emph{$p$th tensor power of $V$}.
	Some authors write $V^{⊗p}$.
	Let $¯v_1,¯v_2…¯v_l$ form a basis of $V$.
	Tensors in $\Vp$ are of the form
	\[∑_{i_1,i_2…i_p∈[l]}a_{i_1i_2\dotsm i_p}¯v_{i_1}⊗¯v_{i_2}⊗\dotsb⊗¯v_{i_p},
		\eqlabel{for:powerbasis}\]
	where $a_{i_1i_2\dotsm i_p}∈𝔽$.
	Same as before, we allow arbitrary vectors to build-up rank-$1$ tensors.
	Thus a tensor in $\Vp$ can be described by
	$∑_ia_iv_{i1}⊗v_{i2}⊗\dotsb⊗v_{ip}$,
	where $a_i∈𝔽$ and $v_{ij}∈V$ are arbitrary.
	The addition is done via matching rank-$1$ tensors.
	The scalar multiplication is distributive.
	The dimension is $\dim(\Vp)=\dim(V)^p=l^p$.
	To avoid confusion, it is worth noting that
	$v_1⊗v_2$ is in general not equal to $v_2⊗v_1$
	unless $v_1$ is a multiple of $v_2$ or $v_2=0$.
	
	\def\WW#1{T^{#1}W}
	Let $\W0$ be $𝔽$;
	let $\W1$ be $W$.
	Let $¯w_1,¯w_2…¯w_k$ form a basis of $W$.
	Let $\Wq$ be the space consisting of formal sums of the form
	\[*∑a_{i_1i_2\dotsm i_q}¯w_{i_1}∧¯w_{i_2}∧\dotsb∧¯w_{i_q},\]*
	where the summation is over all $i_1,i_2…i_q∈[k]$
	such that $1≤i_1<i_2<\dotsb<i_q≤k$.
	And each $¯w_{i_1}∧¯w_{i_2}∧\dotsb∧¯w_{i_q}$
	is an unbreakable, free variable.
	When $q<0$ or $q>k$, the summation is empty,
	so the space is a singleton $𝔽^0=\{0\}$.
	The space becomes interesting after we define the \emph{\WM/}
	\[*Δ：\WW q⟶\Wq\]*
	that sends $¯w_{j_1}⊗¯w_{j_2}⊗\dotsb⊗¯w_{j_q}$ to
	\[*\begin{cases*}
		0∈Λ^qW & if some indices coincide, \\
		(-1)^σ¯w_{i_1}∧¯w_{i_2}∧\dotsb∧¯w_{i_q} & otherwise,
	\end{cases*}\]*
	where $i_1<i_2<\dotsb<i_q$ is the sorted copy of the indices $j_1,j_2…j_q$,
	and $σ$ is the number of swaps used to sort.
	For a sum of several $¯w_{j_1}⊗¯w_{j_2}⊗\dotsb⊗¯w_{j_q}$
	like \cref{for:powerbasis},
	$Δ$ applies to each summand and the images are added together.
	This makes $Δ$ a linear transformation.
	
	Elements of $\Wq$ are also called \emph{tensors}.
	The \WM/ $Δ$ allows us to describe tensors in $\Wq$ more concisely.
	We interpret
	\[*w_1∧w_2∧\dotsb∧w_q\]*
	as
	\[*Δ(w_1⊗w_2⊗\dotsb⊗w_q)∈\Wq,\]*
	where $w_1,w_2…w_q∈W$.
	Then we can use arbitrary vectors in $W$ to describe tensors in $\Wq$:
	What make up $\Wq$ are sums of rank-$1$ tensors
	of the form $∑_ia_iw_{i1}∧w_{i2}∧\dotsb∧w_{iq}$,
	where $a_i∈𝔽$ and $w_{ij}∈W$ are arbitrary.
	The addition is done via matching rank-$1$ tensors.
	The scalar multiplication is distributive.
	The dimension is $\dim(\Wq)=\binom{\dim(W)}q=\binom kq$.
	This syntax has the following two famous characterizations.
	\begin{itemize}
		\item \emph{Multilinearity}.
			It is linear in every of its $w$'s, meaning that\\
			$w_1∧\dotsb∧(w_i+cw_i')∧\dotsb∧w_q$ is equal to\\
			$(w_1∧\dotsb∧w_i∧\dotsb∧w_q)+c(w_1∧\dotsb∧w_i'∧\dotsb∧w_q)$.
		\item \emph{Anti-commutativity}.
			Repetition causes void, meaning that\\
			$w_1∧\dotsb∧w_i∧\dotsb∧w_j∧\dotsb∧w_q=0$ if $w_i=w_j$.\\
			This implies that swapping two $w$'s causes a sign change,\\
			$w_1∧\dotsb∧w_i∧\dotsb∧w_j∧\dotsb∧w_q
				=-w_1∧\dotsb∧w_j∧\dotsb∧w_i∧\dotsb∧w_q$.
	\end{itemize}
	Both multilinearity and anti-commutativity are easily verified.
	Note that tensors in $\Vp$ are also multilinear in the same sense---%
	$v_1⊗\dotsb⊗(v_i+cv_i')⊗\dotsb⊗v_q$ is equal to
	$(v_1⊗\dotsb⊗v_i⊗\dotsb⊗v_q)+c(v_1⊗\dotsb⊗v_i'⊗\dotsb⊗v_q)$.
	This, too, is easy to verify.
	
	Consider the wedge square $Λ^2W$.
	We interpret $w∧W$ as the collection of tensors of the form $∑a_iw∧w_i$,
	that is, the sums where the first component is always~$w$.
	We interpret $W∧w$ as the collection of tensors of the form $∑a_iw_i∧w$,
	which is the same subset as $w∧W$.
	For higher wedge powers, one can interpret
	$w_1∧w_2∧W$, $w_1∧W∧W∧w_2$, $w_1∧W∧w_2∧W∧w_3$, etc.\ similarly.
	It is clear that
	they are subspaces of $Λ^qW$ for the obvious choices of $q$.
	In particular, $W∧W∧\dotsb∧W=Λ^qW$.

\subsection{Generating functions}

	We make use of ordinary generating functions (hereafter \emph{ogf})
	to encode and manipulate series of numbers.
	We utilize that the convolution of two series is encoded by the product,
	namely $∑_n\(∑_{p+q=n}a_pb_q\)x^n=\(∑_pa_px^p\)\(∑_qb_qx^q\)$.
	The shifting of a series is encoded by the multiplication by $x$,
	namely $∑_pa_px^{p+1}=x∑_pa_px^p$.
	
	Two critical types of ogfs are as below:
	A geometric progression $c$, $ca$, $ca^2$, $ca^3$, $\dotsc$
	is encoded by $c/(1-ax)$.
	The binomial coefficients
	$\binom k0$, $\binom k1$, $\binom k2$, $\binom k3$, $\dotsc$
	(varying the bottom while fixing the top) is encoded by $(1+x)^k$.
	The former captures the dimensions of tensor powers.
	The latter captures the dimensions of wedge powers.

讀
\section
					   A Tailor Made Algebraic Foundation   					

\label{sec:tailor}

	Let $𝔽$ be a field.
	Let $V$ be $𝔽^{d-k}$.
	Let $W$ be $𝔽^k$.
	Let $U$ be the direct sum $V⊕W$, isomorphic to $𝔽^d$.
	Let $\Vp$ be the $p$th tensor power of $V$.
	Let $\Wq$ be the $q$th wedge power of $W$.
	
	We work on these two types of spaces:
	One is $\Vp⊗V⊗\Wq$, which is exactly $\V{p+1}⊗\Wq$ (up to associativity).
	The other is $\Vp⊗W⊗\Wq$.
	Note that the direct sum of these two spaces is isomorphic to $\Vp⊗U⊗\Wq$.
	The upcoming diagram depicts the direct sum---%
	there is an inclusion map from top to middle induced by the inclusion $V↪U$,
	and an inclusion map from bottom to middle induced by the inclusion $W↪U$.
	\[\CD[skew 3=]{
		\Vp⊗V⊗\Wq      \\
		\Vp⊗U⊗\Wq\ARDS \\
		\Vp⊗W⊗\Wq      
	}\eqlabel{dia:directsum}\]
	The other two maps are---%
	a projection map from middle to top
	induced by the canonical projection $U↠V$,
	and a projection map from middle to bottom
	induced by the canonical projection $U↠W$.

\subsection{Cowedge-multiplication}

	Another map that is equally pivotal is called \emph{\CM/}
	and defined inductively as follows:
	\begin{align*}
		∇：\Vp⊗\W1 &⟶ \Vp⊗W\\
		ν⊗w_1&⟼ν⊗w_1,\steplabel{for:N0}\\
		∇：\Vp⊗\W2 &⟶ \Vp⊗W⊗\W1\\
		ν⊗w_1∧w_2&⟼ν⊗w_1⊗w_2-ν⊗w_2⊗w_1,\steplabel{for:N1}\\
		∇：\Vp⊗\W{q+1} &⟶ \Vp⊗W⊗\Wq\\
		ν⊗ω∧w_1 &⟼ ∇(ν⊗ω)∧w_1+(-1)^qν⊗w_1⊗ω,\steplabel{for:Nq}
	\end{align*}
	for all $p≥0$, all $q≥2$, all $ν∈\Vp$, all $ω∈\Wq$, and all $w_1,w_2∈W$.
	For tensors of higher ranks,
	the map applies term-wisely and the images are added.
	It can be verified that $∇$ is well-defined and linear.
	(The mapping does not depend on
	representation of a tensor as a sum of rank-$1$ tensors, or
	representation of a rank-$1$ tensor as a product of vectors.)
	
	As an example, $∇$ sends $ν⊗w_1∧w_2∧w_3∧w_4$ to the tensor
	$ν⊗(w_1⊗w_2∧w_3∧w_4
	     -w_2⊗w_1∧w_3∧w_4
	     +w_3⊗w_1∧w_2∧w_4
	     -w_4⊗w_1∧w_2∧w_3)$
	for all $p≥0$, all $ν∈\Vp$, and all $w_1,w_2,w_3,w_4∈W$.
	For rank-$1$ tensors with general $q$, it brings each of the constituent
	vectors in the $\W{q+1}$-segment ($w_1…w_4$ in the above example)
	to the ``front'', and then it assigns alternating signs.
	Note that we could have written $w_1∧w_2∧w_3∧w_4$ as $-w_1∧w_3∧w_2∧w_4$
	or as $w_2∧w_1∧w_4∧w_3$, but the alternating rule gives the same signs.
	
	Convention:
	Latin letters $u,v,w$ represent vectors.
	Specifically, $u∈U$ and $v∈V$ along with $w∈W$.
	Greek letters $ν,ω$ represent tensors.
	In particular, $ν$ (nu) represents a sum of tensors of $v$'s in $V$;
	letter $ω$ (omega) represents a sum of tensors of $w$'s in $W$.

\subsection{Coboundary operators}

	We now define some \emph{differentials}.
	More precisely, they are \emph{coboundary operators},
	which means that they increase the degree/grade of tensors.
	For every $©v©∈V$, define a linear transformation by induction.
	\begin{align*}
		∂^V_©v©：\Wq &⟶ U⊗\Wq\\
		ω &⟼ 0,\steplabel{for:V-}\\
		∂^V_©v©：U⊗\Wq &⟶ \V1⊗U⊗\Wq\\
		u⊗ω &⟼ ©v©⊗u⊗ω,\steplabel{for:V0}\\
		∂^V_©v©：\Vp⊗U⊗\Wq &⟶ \V{p+1}⊗U⊗\Wq\\
		ν⊗u⊗ω &⟼ ∂^V_©v©(ν)⊗u⊗ω+(-1)^pν⊗©v©⊗u⊗ω\steplabel{for:Vp}
	\end{align*}
	for all $p≥1$, all $q≥0$, all $ν∈\Vp$, all $u∈U$, and all $ω∈\Wq$.
	Note that to compute $∂^V_©v©(ν)$ when $ν∈\V1$,
	the argument is included through the chain $ν∈\V1=V↪U≅U⊗\W0$;
	as a result, $∂^V_©v©(ν)=©v©⊗ν∈\V1⊗U⊗\W0$.
	Similarly, for $ν∈\V{p+1}$, place it in $\V{p+1}≅\Vp⊗V↪\Vp⊗U≅\Vp⊗U⊗\W0$.
	For tensors of higher ranks,
	the map applies term-wisely and the images are added.
	It can be verified that $∂^V_©v©$ is well-defined and linear
	(the result does not depend on the representation).
	
	As an example, $∂^V_©v©$ sends $v_1⊗v_2⊗v_3⊗u_4⊗ω$ to the tensor
	$(©v©⊗v_1⊗v_2⊗v_3⊗u_4
	 -v_1⊗©v©⊗v_2⊗v_3⊗u_4
	 +v_1⊗v_2⊗©v©⊗v_3⊗u_4
	 -v_1⊗v_2⊗v_3⊗©v©⊗u_4)⊗ω$
	for all $v_1,v_2,v_3∈V$, all $u_4∈U$, all $q$, and all $ω∈\Wq$.
	For general $p$, it inserts $©v©$ into each of $p+1$ ``gaps''
	in the $\Vp$-segment and then assigns alternating signs.
	The signs are designed such that some other equalities
	in our constructions do not come with convoluted signs.
	For example, it can be shown that $(∂^V_v)^2=0$ as a transformation.
	Elaborately, $∂^V_v\(∂^V_v(ν⊗u⊗ω)\)=0$.
	Also, one can verify that $∂^V_v$ is linear in $v$;
	that is to say, $∂^V_{v+cv'}=∂^V_v+c∂^V_{v'}$ as transformations.
	These two properties together imply that
	$∂^V_v∂^V_{v'}+∂^V_{v'}∂^V_v=(∂^V_{v+v'})^2-(∂^V_v)^2-(∂^V_{v'})^2=0$.
	To sum up, the composition of $∂^V$s is anti-communicative.
	Moreover, since $∂^V_v$ involves operations on the $\Vp$ segment only,
	whether the $\Wq$ segment is included in the argument does not matter,
	\[∂^V_v(ν⊗u⊗ω)=∂^V_v(ν⊗u)⊗ω.\eqlabel{for:V=}\]
	Be aware that \cref{for:V=} is not equal to $∂^V_v(ν)⊗u⊗ω$,
	the difference being $ν⊗v⊗u⊗ω$.
	
	Here are some other transformations.
	For every $w∈W$, define a linear transformation by induction.
	\begin{align*}
		∂^W_©w©：\Vp⊗U &⟶ \Vp⊗U⊗\W1\\
		ν⊗u &⟼ (-1)^pν⊗u⊗©w©,\steplabel{for:W0}\\
		∂^W_©w©：\Vp⊗U⊗\Wq &⟶ \Vp⊗U⊗\W{q+1}\\
		ν⊗u⊗ω &⟼ (-1)^{p+q}ν⊗u⊗ω∧©w©\steplabel{for:Wq}
	\end{align*}
	for all $p≥0$, all $q≥1$, all $ν∈\Vp$, all $u∈U$, and all $ω∈\Wq$.
	For tensors of higher ranks,
	the map applies term-wisely and the images are added.
	It can be verified that $∂^W_w$ is well-defined and linear.
	It can be shown that $(∂^W_w)^2=0$.
	Besides, one can verify that $∂^W_w$ is linear in $w$.
	Symbolically, $∂^W_{w+cw'}=∂^W_w+c∂^W_{w'}$ as linear transformations.
	These two properties together imply
	$∂^W_w∂^W_{w'}+∂^W_{w'}∂^W_w=(∂^V_{w+w'})^2-(∂^W_w)^2-(∂^W_{w'})^2=0$.
	In summary, the composition of $∂^W$s is anti-communicative, too.
	
	For all $v∈V$ and $w∈W$,
	we claim that $∂^V_v$ and $∂^W_w$ anti-commute.
	More formally, $∂^V_v∂^W_w+∂^W_w∂^V_v=0$.
	Let $u∈U$ be $w+v$ in the definition
	\[*∂^U_u≔∂^V_v+∂^W_w.\]*
	Immediately we have
	$(∂^U_u)^2=(∂^V_v)^2+∂^V_v∂^W_w+∂^W_w∂^V_v+(∂^W_w)^2=0$.
	Beyond that, $∂^U_u$ is linear in $u$ because
	$∂^V_v$ and $∂^W_w$ are both linear in their subscripts.
	So we have $∂^U_{u+cu'}=∂^U_u+c∂^U_{u'}$
	and $∂^U_u∂^U_{u'}+∂^U_{u'}∂^U_u=0$.
	
	The following lemma summarizes some properties we met so far.
	We will use them without further referring.
	
	\begin{lem}
		These hold for all $v,v'∈V$, all $w,w'∈W$, all $u,u'∈U$, and all $c∈𝔽$:
		\[*\def\arraystretch{1.44}\begin{matrix}
			&	(∂^V_v)^2=0,
			&	∂^V_{v+cv'}=∂^V_v+c∂^V_{v'},
			&	∂^V_v∂^V_{v'}+∂^V_{v'}∂^V_v=0; \\
			&	(∂^W_w)^2=0,
			&	∂^W_{w+cw'}=∂^W_w+c∂^W_{w'},
			&	∂^W_w∂^W_{w'}+∂^W_{w'}∂^W_w=0; \\
			&	∂^U_v=∂^V_v,
			&	∂^U_w=∂^W_w,
			&	∂^V_v∂^W_w+∂^W_w∂^V_v=0; \\
			&	(∂^U_u)^2=0,
			&	∂^U_u∂^U_{u'}+∂^U_{u'}∂^U_u=0,
			&	∂^U_{u+cu'}=∂^U_u+c∂^U_{u'}.
		\end{matrix}\]*
	\end{lem}
	\begin{proof}
		We prove $(∂^V_v)^2=0$ and $∂^V_v∂^W_w+∂^W_w∂^V_v=0$.
		The rest are routine, if not trivial.
		The proof of the former:
		\begin{align*}
			∂^V_®v®\(∂^V_©v©(ν⊗u⊗ω)\)
			&=	∂^V_®v®\(∂^V_©v©(ν)⊗u⊗ω+(-1)^pν⊗©v©⊗u⊗ω\)\bytag{for:Vp}\\
			&=	∂^V_®v®\(∂^V_©v©(ν)⊗u⊗ω\)+(-1)^p∂^V_®v®\(ν⊗©v©⊗u⊗ω\).
		\end{align*}
		It remains to prove that
		the two terms in the last line cancel each other.
		The first term is
		\[*∂^V_®v®\(∂^V_©v©(ν)⊗u⊗ω\)=∂^V_®v®\(∂^V_©v©(ν)\)⊗u⊗ω
			+(-1)^{p+1}∂^V_©v©(ν)⊗®v®⊗u⊗ω.\bytag{for:Vp}\]*
		The second term is
		\begin{align*}
			∂^V_®v®\(ν⊗©v©⊗u⊗ω\)
			&=	∂^V_®v®(ν⊗©v©)⊗u⊗ω+(-1)^{p+1}ν⊗©v©⊗®v®⊗u⊗ω\bytag{for:Vp}\\
			&=	∂^V_®v®(ν)⊗©v©⊗u⊗ω+(-1)^pν⊗®v®⊗©v©⊗u⊗ω\bytag{for:Vp}\\*
			&\kern1em	+(-1)^{p+1}ν⊗©v©⊗®v®⊗u⊗ω.
		\end{align*}
		All terms except $∂^V_®v®\(∂^V_©v©(ν)\)⊗u⊗ω$ cancel.
		By induction on $ν$'s degree/grade, $∂^V_®v®\(∂^V_©v©(ν)\)$ vanishes.
		(The base case is easy.)
		This confirms $(∂^V_v)^2=0$.
		
		The proof of $∂^V_v∂^W_w+∂^W_w∂^V_v=0$:
		This tensor
		\begin{align*}
			∂^V_®v®\(∂^W_©w©(ν⊗u⊗ω)\)
			&=	(-1)^{p+q}∂^V_®v®(ν⊗u⊗ω∧©w©)\bytag{for:Wq}\\
			&=	(-1)^{p+q}∂^V_®v®(ν⊗u)⊗ω∧©w©\bytag{for:V=}
		\end{align*}
		is the opposite of this tensor
		\begin{align*}
			∂^W_©w©\(∂^V_®v®(ν⊗u⊗ω)\)
			&=	∂^W_©w©\(∂^V_®v®(ν⊗u)⊗ω\)\bytag{for:V=}\\
			&=	(-1)^{p+1+q}∂^V_®v®(ν⊗u)⊗ω∧©w©.\bytag{for:Wq}
		\end{align*}
		This confirms $∂^V_v∂^W_w+∂^W_w∂^V_v=0$.
	\end{proof}
	
	The upcoming two equalities are the keys to the repairing scheme
	and the reason behind the rules of assignments of proper signs.
	They can be verified via expanding all definitions carefully.
	We offer two \emph{proofs}---%
	a proof by induction and a \emph{proof} by example.
	
	\begin{lem}[commutator]\label{lem:homotopy}
		For all $v∈V$, all $w∈W$, all $p,q≥0$, all $ν∈\Vp$, and all $ω∈\Wq$,
		both of these hold:
		\begin{align}
			∂^V_©v©\(∇(ν⊗ω)\)-∇\(∂^V_©v©(ν⊗ω)\) &= (-1)^p∇(ν⊗©v©⊗ω),
				\eqlabel{equ:VDDV}\\
			∂^W_©w©\(∇(ν⊗ω)\)-∇\(∂^W_©w©(ν⊗ω)\) &= (-1)^pν⊗©w©⊗ω.
				\eqlabel{equ:WDDW}
		\end{align}
	\end{lem}
	\begin{proof}[A proof by induction]
		In the following argument, $ν∈\Vp$ and $ω∈\Wq$ along with $w_1∈W$.
		The proof of \cref{equ:VDDV}, the V-part, goes as follows.
		We infer
		\begin{align*}
			&{}	∂^V_©v©\(∇(ν⊗ω∧w_1)\)\steplabel{ten:induVD}\\
			&\qquad	=∂^V_©v©\(∇(ν⊗ω)∧w_1+(-1)^qν⊗w_1⊗ω\)\bytag{for:Nq}\\
			&\qquad	=∂^V_©v©\(∇(ν⊗ω)∧w_1\)+(-1)^q∂^V_©v©\(ν⊗w_1⊗ω\)\\
			&\qquad	=∂^V_©v©\(∇(ν⊗ω)\)∧w_1\bytag{for:V=}\\*
			&\kern3em	+(-1)^q∂^V_©v©(ν)⊗w_1⊗ω+(-1)^{p+q}ν⊗©v©⊗w_1⊗ω.
						\bytag{for:Vp}
		\end{align*}
		We also infer
		\begin{align*}
			&{}	∇\(∂^V_©v©(ν⊗ω∧w_1)\)\steplabel{ten:induDV}\\
			&\qquad	=∇\(∂^V_©v©(ν)⊗ω∧w_1\)\bytag{for:V=}\\
			&\qquad	=∇\(∂^V_©v©(ν)⊗ω\)∧w_1+(-1)^q∂^V_©v©(ν)⊗w_1⊗ω
					\bytag{for:Nq}\\
			&\qquad	=∇\(∂^V_©v©(ν⊗ω)\)∧w_1+(-1)^q∂^V_©v©(ν)⊗w_1⊗ω.\bytag{for:V=}
		\end{align*}
		Subtracting \cref{ten:induVD} by \cref{ten:induDV}, we arrive at
		\begin{align*}
			&\kern-2em	∂^V_©v©\(∇(ν⊗ω∧w_1)\)-∇\(∂^V_©v©(ν⊗ω∧w_1)\)\\
			&{}	=∂^V_©v©\(∇(ν⊗ω)\)∧w_1
				+(-1)^q∂^V_©v©(ν)⊗w_1⊗ω+(-1)^{p+q}ν⊗©v©⊗w_1⊗ω\kern-.75em\\*
			&\kern1em	-∇\(∂^V_©v©(ν⊗ω)\)∧w_1-(-1)^q∂^V_©v©(ν)⊗w_1⊗ω\\
			&{}	=（∂^V_©v©\(∇(ν⊗ω)\)-∇\(∂^V_©v©(ν⊗ω)\)）
					∧w_1+(-1)^{p+q}ν⊗©v©⊗w_1⊗ω\\
			&{}	=(-1)^p∇(ν⊗©v©⊗ω)∧w_1+(-1)^{p+q}ν⊗©v©⊗w_1⊗ω \tag{by IH}\\
			&{}	=(-1)^p∇(ν⊗©v©⊗ω∧w_1).\bytag{for:Nq} 
		\end{align*}
		Here (by IH) means the induction hypothesis---%
		it uses \cref{equ:VDDV} with one less~$q$ (i.e., shorter $ω$).
		The $q=0$ case (the base case) is straightforward and omitted.
		This finishes \cref{equ:VDDV}.
		
		The proof of \cref{equ:WDDW}, the W-part, goes as follows.
		We deduce
		\begin{align*}
			&{}	∂^W_©w©\(∇(ν⊗ω)\)=(-1)^{p+q-1}∇(ν⊗ω)∧©w©.\bytag{for:Wq}
		\end{align*}
		We also deduce
		\begin{align*}
			∇\(∂^W_©w©(ν⊗ω)\)
			&=	(-1)^{p-1+q}∇(ν⊗ω∧©w©)\bytag{for:Wq}\\
			&=	(-1)^{p-1+q}∇(ν⊗ω)∧©w©+(-1)^{p-1}ν⊗©w©⊗ω\bytag{for:Nq}.
		\end{align*}
		Thus we conclude
		\begin{align*}
			&{}	∂^W_©w©\(∇(ν⊗ω)\)-∇\(∂^W_©w©(ν⊗ω)\)\\
			&\qquad	=(-1)^{p+q-1}∇(ν⊗ω)∧©w©
				-(-1)^{p-1+q}∇(ν⊗ω)∧©w©-(-1)^{p-1}ν⊗©w©⊗ω\\
			&\qquad	=(-1)^pν⊗©w©⊗ω.
		\end{align*}
		This completes \cref{equ:WDDW}.
	\end{proof}
	
	Despite that a formal proof of \cref{lem:homotopy} is given and sound,
	here is a \emph{proof} by example:
	The V-part, \cref{equ:VDDV}, goes as follows:
	\begin{align*}
		&{}	∂^V_©v©\(∇(v_1⊗v_2⊗w_3∧w_4)\)\steplabel{ten:examVD}\\
		&\qquad	=∂^V_©v©\(v_1⊗v_2⊗(w_3⊗w_4-w_3⊗w_4)\)\\
		&\qquad	=(©v©⊗v_1⊗v_2-v_1⊗©v©⊗v_2+v_1⊗v_2⊗©v©)⊗(w_3⊗w_4-w_3⊗w_4).
	\end{align*}
	And then
	\begin{align*}
		&{}	∇\(∂^V_©v©(v_1⊗v_2⊗w_3∧w_4)\)\steplabel{ten:examDV}\\
		&\qquad	=∇\((©v©⊗v_1⊗v_2-v_1⊗©v©⊗v_2)⊗w_3∧w_4\)\\
		&\qquad	=(©v©⊗v_1⊗v_2-v_1⊗©v©⊗v_2)⊗(w_3⊗w_4-w_4⊗w_3).
	\end{align*}
	Subtracting \cref{ten:examDV} from \cref{ten:examVD}, we obtain
	\begin{align*}
		&{}	∂^V_©v©\(∇(v_1⊗v_2⊗w_3∧w_4)\)-∇\(∂^V_©v©(v_1⊗v_2⊗w_3∧w_4)\)\\
		&\qquad	=v_1⊗v_2⊗©v©⊗(w_3⊗w_4-w_4⊗w_3)\\
		&\qquad	=(-1)^2∇(v_1⊗v_2⊗©v©⊗w_3∧w_4)
	\end{align*}
	This supports \cref{equ:VDDV}.
	The W-part, \cref{equ:WDDW}, goes as follows:
	\begin{align*}
		&{}	∂^W_©w©\(∇(ν⊗w_1∧w_2∧w_3)\)\steplabel{ten:examWD}\\
		&\qquad	=∂^W_©w©\(ν⊗(w_1⊗w_2∧w_3-w_2⊗w_1∧w_3+w_3⊗w_1∧w_2)\)\\
		&\qquad	=(-1)^{p+2}ν⊗\(w_1⊗w_2∧w_3∧©w©-w_2⊗w_1∧w_3∧©w©\\*
		&\kern9em	+w_3⊗w_1∧w_2∧©w©\)
	\end{align*}
	And then
	\begin{align*}
		&{}	∇\(∂^W_©w©(ν⊗w_1∧w_2∧w_3)\)\steplabel{ten:examDW}\\
		&\qquad	=(-1)^{p+2}∇(ν⊗w_1∧w_2∧w_3∧©w©)\\
		&\qquad	=(-1)^{p+2}ν⊗\(w_1⊗w_2∧w_3∧©w©-w_2⊗w_1∧w_3∧©w©\\*
		&\kern9em	+w_3⊗w_1∧w_2∧©w©-©w©⊗w_1∧w_2∧w_3\).
	\end{align*}
	Subtract \cref{ten:examDW} from \cref{ten:examWD};
	we derive that $∂^W_©w©\(∇(ν⊗w_1∧w_2∧w_3)\)+∇\(∂^W_©w©(ν⊗w_1∧w_2∧w_3)\)
		=(-1)^pν⊗©w©⊗w_1∧w_2∧w_3$.
	This supports \cref{equ:WDDW}.
	
	That is all preparation we need to define and validate moulin code.
	We start an example in the next section.

讀
\section
						 The Special Case \PM s = 4$s=4$						

\label{sec:example}

	In this section, we construct moulin codes with size $s=4$.
	The other parameters are arbitrary but must satisfy $n>d≥k≥s-1=3$.
	We will claim the file format,
	the node configuration, and the repairing rule;
	we will verify them.
	We also estimate $M$, $α$, and $β$ after we define the file format,
	the node configuration, and the repairing rule, respectively.
	In cascade code terminology, this code has mode $μ=s-1=3$.
	In terms of concatenated layered codes,
	this code starts with a root layered code with layer size $s$,
	and the root code is protected by codes with layers of smaller sizes.
	
	Here is the big picture of the spaces we will be working on in this section.
	\[
	\CD[skew 3=-6.5]{
		\iflabor \W4\ARCM & \VW13      & \VW22  & \VW31  & \fi \V4        \\
		\iflabor          & U⊗\W3\ARDS & \VUW12 & \VUW21 & \fi \V3⊗U\ARDS \\
		\iflabor          & W⊗\W3      & \VWW12 & \VWW21 & \fi \V3⊗W      
	}\eqlabel{dia:play4}\]
	See \cref{sec:gallery} for big pictures for other size $s$.
	From there it should be easy to generalize to arbitrary $s$
	(and arbitrary $n,k,d$).
	The big picture contains three types of spaces.
	Spaces in the top row are called \emph{V-spaces}.
	Spaces in the middle row are called \emph{U-spaces}.
	Spaces in the bottom row are called \emph{W-spaces}.
	Spaces in the same column are in the same direct-sum relationship,
	i.e., \cref{dia:directsum}.
	Spaces connected by a south east arrow labeled $∇$
	are in the same \CM/ relationship,
	i.e., \crefrange{for:N0}{for:Nq}.

\subsection{File format and \PM M$M$}

	The file we are to store is seen as a functional $ϕ$
	from the direct sum of all U-spaces to $𝔽$, subject to some parity checks.
	Equivalently, the file $ϕ$ is an element in
	the direct sum of all U-spaces' dual.
	In our $s=4$ example,
	\begin{align*}
		ϕ
		&∈	(U⊗\W3⊕\V1⊗U⊗\W2⊕\V2⊗U⊗\W1⊕\V3⊗U)ˇ\\
		&=	（⨁_{p+q=3}\Vp⊗U⊗\Wq）ˇ≅⨁_{p+q=3}(\Vp⊗U⊗\Wq)ˇ.\steplabel{spa:file4}
	\end{align*}
	For general size $s$, direct sum over $p+q=s-1$.
	
	For downloading and repairing to work, the file $ϕ$ needs to
	satisfy the following kind of parity checks:
	the diagram
	\[\CD[parity 2]{
		\Vp⊗\W{q+1}\drar{∇}\ar{rr}{ϕ}	&&	𝔽	\\&	\Vp⊗W⊗\Wq\urar[']{ϕ}
	}\eqlabel{dia:PC}\]
	commutes for all possible $p≥1$ and $q≥0$ such that $p+q=s-1=3$.
	Wherein, $ϕ$ can be evaluated at a V-space $\Vp⊗\W{q+1}$
	because there is a natural inclusion going downward
	into the U-space $\V{p-1}⊗U⊗\W{q+1}$ and $ϕ$ is well-defined there.
	Similarly, $ϕ$ can be evaluated at a $W$-space $\Vp⊗W⊗\Wq$
	because there is a natural inclusion going upward
	into the U-space $\Vp⊗U⊗\Wq$ and $ϕ$ is well-defined there.
	So a more accurate parity check diagram looks like a twisted pentagon:
	\[*\kern2.5em\CD[skew 3=]{
		\Vp⊗\W{q+1}\dar[c->]\ARCM[',pos=.66]&	𝔽							\\
		\kern-2.5em\V{p-1}⊗U⊗\W{q+1}\urar{ϕ}&	\Vp⊗U⊗\Wq\uar[',pos=.8]{ϕ}	\\
											&	\Vp⊗W⊗\Wq\uar[c->]			
	}\]*
	In terms of tensors, the parity check reads $ϕ(ν⊗ω)=ϕ\(∇(ν⊗ω)\)$
	for all possible $p≥1$ and $q≥0$ such that $p+q=s-1=3$,
	all $ν∈\Vp$, and all $ω∈\W{q+1}$.
	For general size $s$, quantify over $p+q=s-1$.
	
	There are two edge cases to be clarified.
	\emph{Root check}:
	for $p=0$, the diagram degenerates into the following
	because $\W4$ admits no inclusion going downward.
	\[\CD[parity 2]{
		\W4\drar{∇}\ar{rr}{0}	&&	𝔽	\\&	W⊗\W 3\urar[']{ϕ}
	}\eqlabel{dia:RC}\]
	It reads $0=ϕ\(∇(ω)\)$ for all $ω∈\W4$.
	For general size $s$, quantify over $ω∈\W s$.
	The root check is essential in the edge cases of the repairing scheme.
	\emph{Leaf check}:
	for $q=-1$, the diagram degenerates into the following
	because $\V4$ admits no \CM/ going rightward.
	\[\CD[parity 2]{
		\V4\drar{∇}\ar{rr}{ϕ}	&&	𝔽	\\&	0\urar[']{}
	}\eqlabel{dia:LC}\]
	It reads $ϕ(ν)=0$ for all $ν∈\V4$.
	For general size $s$, quantify over $ν∈\V s$.
	The leaf check plays a role as the initial point in the downloading scheme.
	
	Now we can calculate $M$, the file size, as $ϕ$'s degree of freedom.
	It is the total dimension of the U-spaces
	minus the total dimension of parity checks.
	The former is $∑_{p+q=3}d(d-k)^p\binom kq$.
	This quantity is encoded as the $x^4$-\coef/ of
	\[÷{dx(1+x)^k}{1-(d-k)x}.\eqlabel{fun:Uspace}\]
	The latter, the dimension of parity checks, is (no greater than)
	the total dimension of V-spaces, which is $∑_{p+q=4}(d-k)^p\binom kq$.
	This quantity is encoded as the $x^4$-\coef/ of
	\[÷{(1+x)^k}{1-(d-k)x}.\eqlabel{fun:Vspace}\]
	Thus $M$, the dimension of freedom of $ϕ$, is the difference
	$∑_{p+q=3}d(d-k)^p\binom kq-∑_{p+q=4}(d-k)^p\binom kq$.
	This quantity is encoded as the $x^4$-\coef/s of
	\[*÷{dx(1+x)^k}{1-(d-k)x}-÷{(1+x)^k}{1-(d-k)x}
		=÷{(-1+dx)(1+x)^k}{1-(d-k)x}.\copytag{fun:em}\]*
	
	The purpose of encoding dimensions in ogfs
	is to carry parameters in a rather concise form.
	It also helps us find alternative expressions for the same quantity.
	For instance, \cref{fun:em} is equal to
	\[*÷{(-1+dx)(1+x)^k}{1-(d-k)x}=÷{kx(1+x)^k}{1-(d-k)x}-(1+x)^k.\]*
	Thus its $x^4$-\coef/ is also $∑_{p+q=3}k(d-k)^p\binom kq-\binom k4$.
	This expression coincides with \cite[Formula~(3)]{EM19c}.
	Working with ogfs makes it easy to calculate and compare
	different parameters and their relations.
	For general size $s$, the file size $M$ is
	the $x^s$-\coef/ of the same \cref{fun:em}.

\subsection{Node configuration and \PM\U03B1+$α$}\label{sec:node4}

	Let $[n]$ be the set of integers $\{1,2…n\}$;
	they are one-to-one corresponding to the nodes.
	For each $h∈[n]$, the $h$th node selects a \emph{star vector} $u🔑h∈U$.
	The selection of star vectors is such that:
	\itref{Sd} any $d$ star vectors span $U$; and
	\itref{Sk} any $k$ star vectors span $U/V$.
	If we see star vectors as column vectors of a $d$-by-$n$ matrix,
	the first condition says that any $d$ columns form an invertible submatrix.
	The second condition says that
	the first $k$ rows of any $k$ columns form an invertible submatrix.
	For each $h∈[n]$, the $h$th node stores the restrictions
	of the file $ϕ$ to every U-space but $U$ is replaced by $u🔑h$.
	Symbolically, the restriction to this direct sum
	is stored in the $h$th node:
	\begin{align*}
		&{}	\V3⊗u🔑h⊕\V2⊗u🔑h⊗\W1⊕\V1⊗u🔑h⊗\W2⊕u🔑h⊗\W3\steplabel{spa:node4}\\
		&\qquad	=⨁_{p+q=3}\Vp⊗u🔑h⊗\Wq⊆†\cref{spa:file4}†.
	\end{align*}
	For general size $s$, direct sum over $p+q=s-1$.
	
	By the linearity of $ϕ$, nodes do not,
	and should not, store all possible evaluations.
	The $h$th node should choose a basis of \cref{spa:node4}
	and stores the evaluations at that basis.
	Which basis to choose is up to the implementation of the $h$th node
	and is opaque to users and the other nodes.
	When an entity---be it a user or another node---asks
	the $h$th node for an evaluation at a non-basis tensor,
	the $h$th node computes the linear combination on the fly and replies.
	For all intents and purposes, assume that the $h$th node stores
	all evaluations at \cref{spa:node4}, which costs $α$ symbols.
	
	Now we calculate $α$, the dimension of \cref{spa:node4}.
	We know $\dim(\Vp)=(d-k)^p$ and $\dim(\Wq)=\binom kq$.
	So $α$ is $∑_{p+q=3}(d-k)^p\binom kq$.
	This quantity is encoded as the $x^4$-\coef/ of
	\[*÷{x(1+x)^k}{1-(d-k)x}.\copytag{fun:alpha}\]*
	For general size $s$, the node capacity $α$ is
	the $x^s$-\coef/ of the same ogf.

\subsection{Download scheme}

	Without loss of generality, we download from the first $k$ nodes,
	which feature star vectors $u🔑1…u🔑k$.
	We receive the restriction $ϕ↾\V3⊗u🔑1$ from node number $1$;
	we receive the restriction $ϕ↾\V3⊗u🔑2$ from node number $2$;
	and so on and so forth.
	
	Let $\Ws⊆U$ be the subspace spanned by $u🔑1…u🔑k$.
	We can synthesize the restriction of the file to
	$\spa⟨\V3⊗u🔑1…\V3⊗u🔑k⟩=\V3⊗\Ws$ from the given restrictions by linearity.
	Our first step is to study the restriction to the U-space $\V3⊗U$.
	To leap from $\Ws$ to~$U$, recall that $ϕ↾\V3⊗V=ϕ↾\V4$ vanishes;
	this is in virtue of \cref{dia:LC}.
	According to how star vectors are chosen,
	particularly \itref{Sk}, $\Ws∪V$ spans $U$.
	So we can infer $ϕ↾\spa⟨\V3⊗\Ws,\V3⊗V⟩=ϕ↾\V3⊗U$ by linearity.
	This is the first U-space we understand;
	the first step is complete.
	
	Our second step is to study
	the restriction to the U-space $\V2⊗U⊗\W1$.
	We receive the restriction $ϕ↾\V2⊗u🔑h⊗\W1$
	directly from the $h$th node, for each $h∈[k]$.
	We then synthesize the restriction $ϕ↾\V2⊗\Ws⊗\W1$.
	Now we miss the restriction to the V-space $\V2⊗V⊗\W1=\V3⊗\W1$
	in order to meet the second goal $\V2⊗U⊗\W1$.
	To gain $ϕ↾\V3⊗\W1$, recall that we just recovered $ϕ↾\V3⊗U$
	in the previous paragraph, so we know $ϕ↾\V3⊗W$ as a further restriction.
	Invoke the \CM/ relation, i.e., \cref{dia:PC};
	we deduce $ϕ↾\V3⊗\W1$.
	Now we finish recovering $ϕ↾\V2⊗U⊗\W1$;
	the second step is complete.
	
	For the third U-space $\V1⊗U⊗\W2$ and what follows,
	we continue with an induction argument.
	In each step, the induction hypothesis is that we know $ϕ↾\V{p+1}⊗W⊗\Wq$.
	From the induction hypothesis and the \CM/ relation
	we infer $ϕ↾\Vp⊗V⊗\W{q+1}$.
	We also receive $ϕ↾\Vp⊗\Ws⊗\W{q+1}$ directly from the first $k$ nodes.
	Since $\Ws$ and $V$ span $U$, we deduce $ϕ↾\Vp⊗U⊗\W{q+1}$.
	And then we proceed to the next $p$ and $q$
	(decreasing $p$ by $1$ and increasing $q$ by~$1$).
	The induction lets us study all U-spaces, one by one,
	and hence we comprehend the file~$ϕ$.
	See \cref{fig:download} for the recovery schedule in its entirety.
	
	\begin{figure}
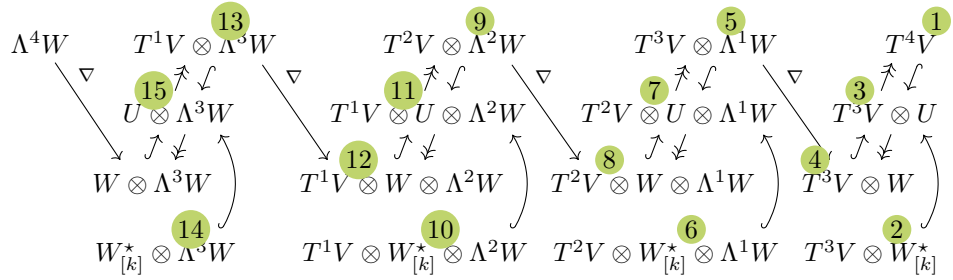

		\def\VSW#1#2{\V#1⊗\Ws⊗\W#2\ARST}
		\def\ARST{\ar[start anchor={north east},end anchor=south east,
				xshift=-1em,bend right,c->,overlay]{uu}}
		$$\CD[skew 3=,/tikz/row 4/.style={xshift=1.5em},
			execute at end picture={
				\iflabor
				\tikzset{nodes={overlay,circle,inner sep=1,fill=Citron}}
				\def\M{\tikzcdmatrixname}
				\def\step{1}\def\stepstep{\xdef\step{\the\numexpr\step+1}}
				\foreach\c in{5,...,2}{
					\draw(\M-1-\c.45)node{$\step$};\stepstep
					\draw(\M-4-\c.45)node{$\step$};\stepstep
					\draw(\M-2-\c.135)node{$\step$};\stepstep
					\ifnum\step<16
					\draw(\M-3-\c.150)node{$\step$};\stepstep
					\fi
				}
				\fi
			}
		]{
			\iflabor \W4\ARCM
			         & \VW13        & \VW22   & \VW31  & \fi \V4          \\
			\iflabor & U⊗\W3\ARDS   & \VUW12  & \VUW21 & \fi \V3⊗U\ARDS   \\
			\iflabor & W⊗\W3        & \VWW12  & \VWW21 & \fi \V3⊗W        \\
			\iflabor & \Ws⊗\W3\ARST & \VSW12  & \VSW21 & \fi \V3⊗\Ws\ARST
		}$$\caption{
			Step by step downloading scheme.
			Step $1$ is the leaf check.
			Step $4j+2$ is receiving from nodes.
			Step $4j+3$ is (Sk).
			Step $4j+4$ is restricting.
			Step $4j+5$ is induction hypothesis and \CM/.
		}\label{fig:download}
	\end{figure}

\subsection{Repair scheme and \PM\U03B2+$β$}\label{sec:repair4}

	Repairing requires some extra preparation.
	Recall that $U≅V⊕W$.
	Thus every star vector $u🔑f$ admits a decomposition
	$u🔑f=v🔑f+w🔑f$ for some $v🔑f∈V$ and some $w🔑f∈W$.
	We will use $u🔑f$, $v🔑f$, and $w🔑f$ to define the coboundary operators.
	And then we specify the repairing rule built upon them.
	We verify the repairing rule in two ways.
	Lastly we calculate $β$.
	
	For each $f∈[n]$, define an alias $∂^V_f≔∂^V_{v🔑f}$.
	To expand the inductive definition:
	\begin{align*}
		∂^V_f(ω) &= 0,\\
		∂^V_f(u⊗ω) &= v🔑f⊗u⊗ω,\\
		∂^V_f(ν⊗u⊗ω) &= ∂^V_f(ν)⊗u⊗ω+(-1)^pν⊗v🔑f⊗u⊗ω
	\end{align*}
	for all $p≥1$, all $q≥0$, all $ν∈\Vp$, all $u∈U$, and all $ω∈\Wq$.
	Recall that $(∂^V_f)^2=0$ and $∂^V_f(ν⊗u⊗ω)=∂^V_f(ν⊗u)⊗ω$;
	we need them later.
	Here is the other family of coboundary operators.
	For each $f∈[n]$, define an alias $∂^W_f≔∂^W_{w🔑f}$.
	To expand the inductive definition:
	\begin{align*}
		∂^W_f(ν⊗u) &= (-1)^pν⊗u⊗w🔑f,\\
		∂^W_f(ν⊗u⊗ω) &= (-1)^{p+q}ν⊗u⊗ω∧w🔑f
	\end{align*}
	for all $p≥0$, all $q≥1$, all $ν∈\Vp$, all $u∈U$, and all $ω∈\Wq$.
	Recall that $(∂^W_f)^2=0$ and $∂^V_f∂^W_f+∂^W_f∂^V_f=0$.
	For each $f∈[n]$, define an alias
	\[*∂^U_f≔∂^U_{u🔑f}=∂^V_f+∂^W_f.\]*
	Remember that $(∂^U_f)^2=0$ and, equivalently,
	\[∂^U_f∂^W_f+∂^U_f∂^V_f=0.\eqlabel{for:U+}\]
	
	When the $f$th node fails,
	the system commands $d$ nodes to help repair it.
	In our repairing scheme, the help message from one helper node
	does not depend on the identity of the other helper nodes
	($𝘚^ℋ_{h→f}$ depends on $h,f$ but not on $ℋ、\{h\}$).
	In other words, it suffices to specify
	what a node should send out if it is commanded to help the $f$th.
	The $h$th node, if signaled,
	sends the failing node restrictions to these subspaces of coboundaries
	\[∂^U_f(\Vp⊗u🔑h⊗\Wq)\eqlabel{spa:help}\]
	for all possible $p,q≥0$ such that $p+q=s-2=2$.
	For general size~$s$, quantify over $p+q=s-2$.
	For repair to work, there are four statements we need to go through:
	\itlabel{Ra} What a helper node sends out is a function of its content.
	\itlabel{Rb} There is a repairing rule that uses nothing but help messages.
	\itlabel{Rc} The rule covers all symbols the failing node used to hold.
	\itlabel{Rd} The rule repairs correctly.
	Cf.\ \cref{equ:d}.
	And there is one quantity to calculate: $β$.
	
	We now go through \itref{Ra}--\itref{Rd}.
	\itref{Ra} is straightforward as \cref{spa:help} is a subspace
	of a direct sum $\V{p+1}⊗u🔑h⊗\Wq⊕\Vp⊗u🔑h⊗\W{q+1}$
	whose direct summands are both contained in \cref{spa:node4}.
	For \itref{Rb}--\itref{Rd}, we claim the repairing rule:
	the failing node regains the evaluation at $ν⊗u🔑f⊗ω$ via computing
	the left-hand side of \cref{equ:repair} in this theorem.
	
	\begin{thm}[repairing rule]\label{thm:repair}
		For all possible $p,q≥0$ such that $p+q=s-2$,
		all $ν∈\Vp$, and all $ω∈\Wq$,
		\[ϕ（∂^U_f\(∇(ν⊗ω)\)）-ϕ（∂^U_f\(ν⊗ω\)）
			=(-1)^pϕ(ν⊗u🔑f⊗ω)\eqlabel{equ:repair}\]
		Note that for the $p=0$ case,
		$ϕ\(∂^U_f(ω)\)=0$ due to the root check, \cref{dia:RC}.
	\end{thm}
	
	\itref{Rb} is equivalent to whether the tensor in the right-hand side of
	\cref{equ:repair} exhausts $\Vp⊗u🔑f⊗\Wq$, which it does.
	\itref{Rc} is equivalent to whether the help messages contain
	the evaluations on the left-hand side of \cref{equ:repair}.
	To answer that, notice that the failing node receives the restriction
	to $∂^U_f(\Vp⊗u🔑h⊗\Wq)$ for a list of $d$ star vectors.
	According to \itref{Sd}, these star vectors span $U$.
	Therefore, the failing node is able to synthesize the restriction to
	$∂^U_f(\Vp⊗U⊗\Wq)$, which forms the entire coboundary.
	Hence \itref{Rc} holds.
	
	We are left with \itref{Rd} and $β$.
	Statement \itref{Rd} is equivalent to the correctness of \cref{equ:repair}.
	\Cref{equ:repair} (\cref{thm:repair}) holds because
	\begin{align*}
		&\kern-1em	ϕ（∂^U_f\(∇(ν⊗ω)\)）-ϕ（∂^U_f\(ν⊗ω\)）\\
		&\quad	=ϕ（∂^U_f\(∇(ν⊗ω)\)）-ϕ（∇\(∂^U_f(ν⊗ω)\)）\bytag{dia:PC}\\
		&\quad	=ϕ\(∂^V_f∇(ν⊗ω)+∂^W_f∇(ν⊗ω)\)-ϕ\(∇∂^V_f(ν⊗ω)+∇∂^W_f(ν⊗ω)\)\\
		&\quad	=ϕ\(∂^V_f∇(ν⊗ω)-∇∂^V_f(ν⊗ω)\)+ϕ\(∂^W_f∇(ν⊗ω)-∇∂^W_f(ν⊗ω)\)\\
		&\quad	=ϕ\((-1)^p∇(ν⊗v🔑f⊗ω)\)+ϕ\((-1)^pν⊗w🔑f⊗ω\)
				\steplabel{equ:homotopy}\\ 
		&\quad	=(-1)^pϕ(ν⊗v🔑f⊗ω\)+(-1)^pϕ(ν⊗w🔑f⊗ω)\bytag{dia:PC}\\
		&\quad	=(-1)^pϕ(ν⊗u🔑f⊗ω).
	\end{align*}
	\Cref{equ:homotopy} is a consequence of
	\begin{align}
		∂^V_f\(∇(ν⊗ω)\)-∇\(∂^V_f(ν⊗ω)\) &= (-1)^p∇(ν⊗v🔑f⊗ω),\\
		∂^W_f\(∇(ν⊗ω)\)-∇\(∂^W_f(ν⊗ω)\) &= (-1)^pν⊗w🔑f⊗ω.
	\end{align}
	They are proven in \cref{lem:homotopy}.
	Now \itref{Rd} (\cref{thm:repair}) is complete.

\subsubsection{The bandwidth $β$}\label{sec:bandwidth4}

	Now we calculate $β$;
	$β$ is the sum of the dimension of the coboundaries
	$∂^U_f(\Vp⊗u🔑h⊗\Wq)$ for all $p+q=s-2=2$.
	To (over)estimate this dimension, let $W^⟂_f⊆W$ be a complement
	of $\spa⟨w🔑f⟩$ in $W$, that is, $W^⟂_f⊕\spa⟨w🔑f⟩≅W$.
	We claim a lemma.
	
	\begin{lem}[compress help]\label{lem:bandwidth}
		Both
		\[∑_{p+q=s-2}∂^U_f(\Vp⊗u🔑h⊗\Wq)⊆∑_{p+q=s-2}∂^U_f(\Vp⊗u🔑h⊗\Wq^⟂_f)
			\eqlabel{con:nof}\]
		and
		\[\dim（†right-hand side of \eqref{con:nof}†）
			≤∑_{p+q=s-2}(d-k)^p\bi{k-1}q\eqlabel{ine:lessk}\]
		hold.
	\end{lem}
	\begin{proof}
		To show \cref{ine:lessk}, one realizes that
		each $\Vp⊗u🔑h⊗\Wq^⟂_f$ is of dimension $(d-k)^p\binom{k-1}q$.
		Now sum over $p+q=s-2$.
		\Cref{ine:lessk} is done.
		
		To show \cref{con:nof}, it suffices to show that every coboundary
		of this form is in the right-hand side: $∂^U_f(ν⊗u🔑h⊗ω)$ for
		all possible $p,q≥0$ such that $p+q=s-2$, all $ν∈\Vp$, and all $ω∈\Wq$.
		To do so, we first decompose the $\Wq$-segment
		of the argument $ν⊗u🔑h⊗ω$ according to
		the direct sum decomposition $W≅W^⟂_f⊕\spa⟨w🔑f⟩$.
		In details, $ω$ can be rewritten as $ω^♮+ω^♭∧w🔑f$
		for some $ω^♮∈\Wq^⟂_f$ and some $ω^♭∈\W{q-1}^⟂_f$.
		Clearly $∂^U_f(ν⊗u🔑h⊗ω^♮)$ lies
		in the right-hand side of \cref{con:nof}.
		It remains to show that so does $∂^U_f(ν⊗u🔑h⊗ω^♭∧w🔑f)$.
		Rewrite it
		\begin{align*}
			∂^U_f(ν⊗u🔑h⊗ω^♭∧w🔑f)
			&=	(-1)^{q-1}∂^U_f\(∂^W_f(ν⊗u🔑h⊗ω^♭)\)\bytag{for:Wq}\\
			&=	(-1)^q∂^U_f\(∂^V_f(ν⊗u🔑h⊗ω^♭)\).\bytag{for:U+}
		\end{align*}
		Argument $∂^V_f(ν⊗u🔑h⊗ω^♭)$ is in $\V{p+1}⊗u🔑h⊗\W{q-1}^⟂_f$,
		hence $∂^U_f\(∂^V_f(ν⊗u🔑h⊗ω^♭)\)$ is in the right-hand side.
		This completes \cref{con:nof}.
	\end{proof}
	
	Finally we conclude that $β$ is (at most)
	the right-hand side of \cref{ine:lessk} plugged in $s=4$.
	It is the $x^4$-\coef/ of
	\[*÷{x^2(1+x)^{k-1}}{1-(d-k)x}.\copytag{fun:beta}\]*
	For general size $s$, it is \cref{for:beta}
	and the $x^s$-\coef/ of \cref{fun:beta}.
	This demonstrates the proof of \cref{thm:moulin,pro:ogf} for the $s=4$ case.
	
	Remark:
	\Cref{con:nof} holds with equality due to $W^⟂_f⊆W$.
	Supposedly, \cref{ine:lessk} holds with equality.
	Otherwise the true $β$ will be less than the estimate above,
	which implies that moulin code will be better than cascade code.
	That \cref{ine:lessk} holds with equality also means that
	the right-hand side of \cref{con:nof} is a direct sum.

讀
\section
							  General Code Summary  							

\label{sec:general}

	This section summarizes the general moulin construction.
	First, \cref{dia:play4} extends to a larger diagram.
	\[\let\ARCN\ARCM\let\ARCM\relax
	\CD[skew 3=-7.5]{
		\W s\ARCN &\VW1{s-1}      &⋯\ARCN & \VW{p+1}q &⋯\ARCN &\V s           \\
		          &U⊗\W{s-1}\ARDS &⋯      & \VUW pq   &⋯      &\V{s-1}⊗U\ARDS \\
		          &W⊗\W{s-1}      &⋯      & \VWW pq   &⋯      &\V{s-1}⊗W      
	}\eqlabel{dia:play}\]
	Spaces in a column from a direct-sum relationship,
	i.e., \cref{dia:directsum}.
	Spaces connected by a $∇$-arrow form a \CM/ relationship,
	i.e., \crefrange{for:N0}{for:Nq}.
	See \cref{sec:gallery} for more instances of diagrams for various~$s$.

\subsection{File format and \PM M$M$}

	The file is in the following space
	\[*ϕ∈（⨁_{p+q=s-1}\Vp⊗U⊗\Wq）ˇ≅⨁_{p+q=s-1}(\Vp⊗U⊗\Wq)ˇ.\]*
	Cf. \cref{spa:file4}.
	The file satisfies some parity checks.
	The general check is $ϕ(ν⊗ω)=ϕ\(∇(ν⊗ω)\)$ for all possible $p≥1$ and $q≥0$
	such that $p+q=s-1$, all $ν∈\Vp$, and all $ω∈\W{q+1}$.
	Cf.\ \cref{dia:PC}.
	The root check is $0=ϕ\(∇(ω)\)$ for all $ω∈\W s$.
	Cf.\ \cref{dia:RC}.
	The leaf check is $ϕ(ν)=0$ for all $ν∈\V s$.
	Cf.\ \cref{dia:LC}.
	
	The file size $M$ is the total dimension of the U-spaces
	minus the total dimension of the V-spaces.
	The former is
	\[*∑_{p+q=s-1}d(d-k)^p\bi kq=[x^s]÷{dx(1+x)^k}{1-(d-k)x}.\]*
	Cf.\ \cref{fun:Uspace}.
	The latter is
	\[*∑_{p+q=s}(d-k)^p\bi kq=[x^s]÷{(1+x)^k}{1-(d-k)x}.\]*
	Cf.\ \cref{fun:Vspace}.
	The difference is \cref{for:em}, the $x^s$-\coef/ of \cref{fun:em}.

\subsection{Node configuration and \PM\U03B1+$α$}

	For every $h∈[n]$, the $h$th node selects $u🔑h∈U$.
	The selection is such that:
	\itlabel{Sd} any $d$ star vectors span $U$; and
	\itlabel{Sk} any $k$ star vectors span $U/V$.
	For every $h∈[n]$, the $h$th node stores
	the restriction of $ϕ$ to this subspace
	\[*⨁_{p+q=s-1}\Vp⊗u🔑h⊗\Wq⊆⨁_{p+q=s-1}\Vp⊗U⊗\Wq.\]*
	Cf.\ \cref{spa:node4}.
	The node capacity $α$ is the dimension of the left-hand side.
	It is \cref{for:alpha}, which is the $x^s$-\coef/ of \cref{fun:alpha}.

\subsection{Download scheme}

	\Cref{fig:download} summarizes the argument better than any text does.
	Text:
	What we download is $ϕ↾\Vp⊗\Ws⊗\W{q+1}$
	for all $p≥0$ and $q≥-1$ such that $p+1+q=s-1$,
	where $\Ws$ is the span of the key vectors of the nodes we download from.
	The base case is, from $ϕ↾\V{s-1}⊗\Ws⊗\W0$ (we downloaded this),
	$ϕ↾\V s=ϕ↾\V{s-1}⊗V$ (the leaf check), and \itref{Sk},
	we acquire $ϕ↾\V{s-1}⊗U⊗\W0$.
	The induction hypothesis is we know $ϕ↾\V{p+1}⊗U⊗\Wq$ for some $q≥0$.
	To proceed, use $ϕ↾\V{p+1}⊗W⊗\Wq$ and the parity check
	to infer $ϕ↾\Vp⊗V⊗\W{q+1}$.
	Use what we downloaded $ϕ↾\Vp⊗\Ws⊗\W{q+1}$ and \itref{Sk}
	to acquire $ϕ↾\Vp⊗U⊗\W{q+1}$.
	The latest restriction is the induction hypothesis for the next step.
	Therefore, we can acquire $ϕ↾\Vp⊗U⊗\W{q+1}$ for all $p+q+1=s-1$.
	Hence we can comprehend $ϕ$.
	
	The key idea in the downloading scheme can be rephrased as follows.
	
	\begin{thm}[downloading scheme]
		For any subspace $W^⋆⊆U$ such that $V+W^⋆=U$,
		\[*\operatorname{image}(\operatorname{id}-∇)+⨁_{p+q=s-1}\Vp⊗W^⋆⊗\Wq
			=⨁_{p+q=s-1}\Vp⊗U⊗\Wq.\]*
	\end{thm}

\subsection{Repair scheme and \PM\U03B2+$β$}

	The explanation in \cref{sec:repair4} applies to general $s$.
	Recap:
	The help message from the $h$th node is the restriction to
	\[*∂^U_f(\Vp⊗u🔑h⊗\Wq)\copytag{spa:help}\]*
	for all possible $p,q≥0$ such that $p+q=s-2$.
	The failing node regain $ϕ(ν⊗u🔑f⊗ω)$ by computing
	\[*ϕ（∂^U_f\(∇(ν⊗ω)\)）-ϕ（∂^U_f\(ν⊗ω\)）=(-1)^pϕ(ν⊗u🔑f⊗ω)
		\copytag{equ:repair}\]*
	for all $ν∈\Vp$, and all $ω∈\Wq$.
	
	To calculate $β$, we recall
	\[*∑_{p+q=s-2}∂^U_f(\Vp⊗u🔑h⊗\Wq)⊆∑_{p+q=s-2}∂^U_f(\Vp⊗u🔑h⊗\Wq^⟂_f).
		\copytag{con:nof}\]*
	That the containment holds is proven as part of \cref{lem:bandwidth}.
	Abstract summary: for every $ν⊗u⊗ω$, the $ω$-segment can be decomposed into
	a $w🔑f$-free part ($ω^♮$) and an incident part ($ω^♭∧w🔑f$).
	The former is in the right-hand side.
	The latter is in the image of $∂^W_f$.
	And then we rewrite $∂^U_f∂^W_f$ as $-∂^U_f∂^V_f$
	to show that it is again in the right-hand side.
	The bandwidth $β$ is the dimension of the right-hand side,
	which is (bounded from above by) \cref{for:beta}
	and the $x^s$-coefficient of \cref{fun:beta}.
	
	At this point, we finish the proof of \cref{pro:ogf}.
	We almost finish proving \cref{thm:moulin}
	except that the field $𝔽$ and its size is unclear.
	It is covered in the next subsection.

\subsection{Field size}\label{sec:field}

	The last piece of \cref{thm:moulin} concerns the field size $\abs{𝔽}$.
	Throughout the paper, there was only one factor
	that (potentially) limits which field $𝔽$ can be.
	That is the conditions \itref{Sd} and \itref{Sk}.
	\itref{Sd} states that we must be able to find $n$ vectors in $𝔽^d$
	such that any $d$ of them are linearly independent.
	\itref{Sk} states that, when projected onto a $k$-dimensional
	quotient space $U/V$, any $k$ of them are linearly independent.
	
	To fulfill the nested requirements, consider Reed--Solomon codes.
	Let every node choose a unique star element $a🔑h∈𝔽$;
	and let $u🔑h$ be $[1\ a🔑h\ (a🔑h)^2\ ⋯\ (a🔑h)^{d-1}]$.
	Then \itref{Sd} is satisfied.
	We choose $V$ to be spanned by the last $d-k$ components of $𝔽^d$.
	Thus $w🔑h$ becomes the first $k$ components of $u🔑h$,
	which is $[1\ a🔑h\ (a🔑h)^2\ ⋯\ (a🔑h)^{k-1}]$.
	Hence \itref{Sk} is satisfied, too.
	Since there are $\abs{𝔽}$ distinct elements to choose from,
	a field size of $n$ is sufficient.
	This finishes the proof of \cref{thm:moulin}.
	
	Finding vectors that fulfill \itref{Sk} alone
	is equivalent to finding an $[n,k]$-\MDS/ code.
	For the search of \MDS/ codes, it is known that
	if $\abs{𝔽}<k$ then $n≤k+1$, where the equality is achieved by
	the identity matrix augmented by the all-one vector.
	For if $\abs{𝔽}≥k$, the \MDS/ conjecture states that $n≤\abs{𝔽}+2$.
	In conclusion, either
	we focus on the $n=k+1$ case and enjoy a field size as small as $2$,
	or we enlarge the field linearly along with $n$.
	The latter strategy is what \cref{thm:moulin} does.
	Note that the lower bound $\abs{𝔽}≥n-2$ is effective
	within the territory of our code and cascade codes.
	It does not say anything about the field size
	of other possible ERRC constructions.
	
	Finally, we utilize an augmented identity matrix
	to construct ERRCs for the $n=d+1=k+1$ case.
	The precise statement is made into the next proposition.
	This code happens to be a layered code \cite{TSAVK15}.
	
	\begin{pro}[layered code]
		For any integers $k$ and $s$ such that $k≥s-1≥1$,
		there exists a $(k+1,k,k,α,β,M)$-ERRC with
		\begin{align*}
			α &= \bi k{s-1},\\
			β &= \bi{k-1}{s-2},\rlap{\qquad and}\\
			M &= k\bi k{s-1}-\bi ks
		\end{align*}
		over any field.
	\end{pro}

讀
\section
						    Repair Multiple Failures    						

\label{sec:multiple}

	Assume general $n$, $k$, $d$, and $s$;
	that is, $n-1≥d≥k≥s-1≥1$.
	There are two models that measure the cost of repairing multiple failures.
	We elaborate on them in the next subsection.
	We focus on the centralized model in the future subsections.
	As an example, we first attempt to repair two failing nodes.
	And then we generalize to $c≤n-d$ failures.
	A similar analysis was conducted for the $k=d$ case in \cite{EM19d}.
	Throughout the section, $h$ means any of the helper indices.

\subsection{Centralized and cooperative models}

	In general, failures separate in time.
	But there may be circumstances where multiple nodes fail at once.
	\Cref{dfn:regenerate} does not cover the case
	when there are more nodes to be repaired.
	What \cref{dfn:regenerate} guarantees is that,
	A, so far as there are $k$ healthy nodes left, the file is safe.
	B, if there are $d$ healthy nodes left,
	one may call the repairing protocol for each and every failing node.
	This does not capture how efficient the repairing can be done.
	For that, two definitions are made in \cite{CJMRS13, SH13},
	and related in \cite{YB19}.
	
	\begin{dfn}\label{dfn:cooperative}
		Let there be $c$ failing nodes and $d$ helper nodes.
		The \emph{centralized (total) bandwidth} $γ\ce(c,d)$ is the total number
		of symbols the $d$ helper nodes send to a central agent who,
		after gathering all help messages, will repair the failures.
		The \emph{cooperative (total) bandwidth} $γ\co(c,d)$ is
		how many symbols are sent over the network,
		from a helping node or a failing one, that contribute to repairing.
	\end{dfn}
	
	See \cref{fig:bee} for illustration.
	Note that we may as well normalize the total bandwidth
	by the number of helper, or failing, nodes.
	One reason for doing so is to compare total bandwidths
	with the repair bandwidth in \cref{dfn:regenerate}.
	Particularly, $γ\ce(1,d)=γ\co(1,d)=dβ$.
	We analyse our code's performance under the centralized model.
	In doing so, our focus is on individual helper nodes.
	Hence we use the per-helper bandwidth $β_c≔γ\ce(c,d)/d$ to benchmark.
	In particular, $β_1=β$.
	
	\begin{figure}
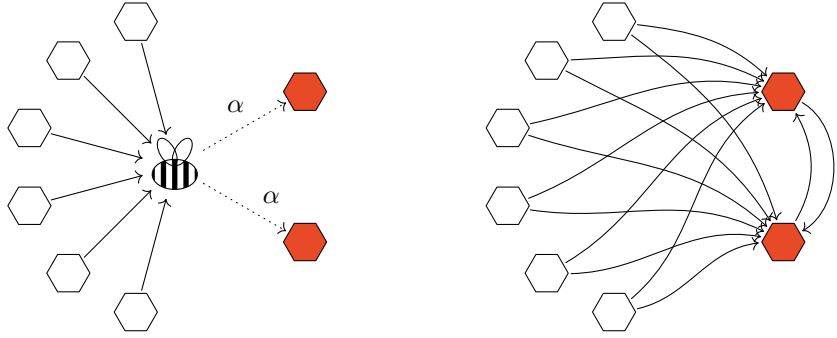

		$$\tikz[shorten both=1]{
			\draw(0,0)node[circle,inner sep=8](bee){};
			\draw foreach\d in{1,...,6}{
				(\d*30+75:2)苯(D\d){}(D\d)edge[->](bee)
			};
			\iflabor
			\begin{scope}
				\draw(-.1,.2)ellipse[x radius=.1,y radius=.2,rotate=30];
				\draw(.1,.2)ellipse[x radius=.1,y radius=.2,rotate=-30];
				\draw[clip](0,-.1)ellipse(.3 and .2);
				\draw[line width=2,scale=.15]
					foreach\x in{-3,...,3}{(\x,-2)--+(0,4)};
			\end{scope}
			\draw foreach\c in{1,...,2}{
				(\c*60-90:2)苯(C\c)[fill=UCO]{}
				(bee)edge[->,dotted]node[pos=.6,auto]{$α$}(C\c)
			};
			\fi
		}\hskip6em
		\tikz[shorten both=1]{
			\iflabor
			\draw foreach\d in{1,...,6}{
				(\d*30+75:2)苯(D\d){}(D\d)
				edge[->,bend left=21-6*\d](C1)edge[->,bend left=21-6*\d](C2)
			};
			\draw foreach\c in{1,2}{(\c*60-90:2)苯(C\c)[fill=UCO]{}};
			\draw(C1)edge[->,bend left=-30](C2)(C2)edge[->,bend left=60](C1);
			\fi
		}$$
		\caption{
			To the left: centralized model.
			Healthy nodes send help messages to an agent.
			The total bandwidth $γ\ce(2,6)$ is
			the number of symbols passing solid lines.
			To the right: cooperative model.
			Healthy nodes send help messages directly to the failing nodes,
			while the latter can help each other.
			The total bandwidth $γ\co(2,6)$ is
			the number of symbols passing solid lines.
		}\label{fig:bee}
	\end{figure}

\subsection{Repair two failures}

	To repair two nodes, $f$ and $g$,
	it suffices to send out the restrictions to coboundaries
	of the form $∂^U_f(\Vp⊗u🔑h⊗\Wq)$ and $∂^U_g(\Vp⊗u🔑h⊗\Wq)$.
	The main issue here is to count the dimension (aka the total bandwidth).
	
	Recall the motivation of $W^⟂_f$ is to ``avoid'' $w🔑f$.
	Let $W^⟂_{fg}$ be a subspace avoiding $w🔑f$ and $w🔑g$.
	That is, $W^⟂_{fg}$ is such that $W^⟂_{fg}⊕\spa⟨w🔑f,w🔑g⟩≅W$.
	Redefine $W^⟂_f$ as $W^⟂_{fg}⊕\spa⟨w🔑g⟩$;
	note that $W^⟂_f⊕\spa⟨w🔑f⟩≅W$ still holds.
	We next claim a containment
	\begin{align*}
		&{}	∑_{p+q=s-2}∂^U_f(\Vp⊗u🔑h⊗\Wq)+∑_{p+q=s-2}∂^U_g(\Vp⊗u🔑h⊗\Wq)\\
		&\qquad	⊆∑_{p+q=s-2}∂^U_f(\Vp⊗u🔑h⊗\Wq^⟂_f)
			+∑_{p+q=s-2}∂^U_g(\Vp⊗u🔑h⊗\Wq^⟂_{fg})\steplabel{con:nofg}
	\end{align*}
	and an inequality
	\[\dim（†right-hand side of \eqref{con:nofg}†）
		≤∑_{p+q=s-2}(d-k)^p\bi{k-1}q+(d-k)^p\bi{k-2}q.\eqlabel{ine:2lessk}\]
	
	Overall strategy:
	Based on the idea of \cref{sec:bandwidth4},
	we again squeeze the helping subspace into a
	(seemingly smaller) subspace followed by counting the dimension.
	By squeezing, we look at the $\Wq$-segment and attempt
	to remove $w🔑f$ and $w🔑g$ as thoroughly as possible.
	Removing one star vector is done in \cref{sec:bandwidth4}.
	The main challenge here is to remove $w🔑f$ from the argument of $∂^U_g$.
	(Notice the incoherence in indices.)
	
	Actual verification:
	\Cref{ine:2lessk} is as straightforward as \cref{ine:lessk}.
	\Cref{con:nofg}, on the other hand, is more involved.
	It is a consequence of the following three containments---%
	\cref{con:nof}, \cref{con:nof} with all $f$'s replaced by $g$,
	and the one in the lemma below.
	
	\begin{lem}[compress joint help]\label{lem:compress}
		\begin{align*}
			&{}	∑_{p+q=s-2}∂^U_g(\Vp⊗u🔑h⊗\Wq^⟂_g)\\
			&\qquad	⊆∑_{p+q=s-2}∂^U_f(\Vp⊗u🔑h⊗\Wq)
				+∑_{p+q=s-2}∂^U_g(\Vp⊗u🔑h⊗\Wq^⟂_{fg}).\steplabel{con:nog}
		\end{align*}
	\end{lem}
	\begin{proof}
		Let $∂^U_g(ν⊗u🔑h⊗ω)$ be in the left-hand side, where $ω∈\Wq^⟂_g$.
		Now we decompose $ω$ according to $W^⟂_g≅W^⟂_{fg}⊕\spa⟨w🔑f⟩$.
		That is, $ω=ω^♮+ω^♭∧w🔑f$
		for some $ω^♮∈\Wq^⟂_{fg}$ and some $ω^♭∈\W{q-1}^⟂_{fg}$.
		Clearly $∂^U_g(ν⊗u🔑h⊗ω^♮)$ is in the right-hand side.
		It remains to fit $∂^U_g(ν⊗u🔑h⊗ω^♭∧w🔑f)$ in.
		To do so, rewrite the coboundary
		\begin{align*}
			∂^U_g(ν⊗u🔑h⊗ω^♭∧w🔑f)
			&=	(-1)^{q-1}∂^U_g∂^W_f(ν⊗u🔑h⊗ω^♭)\\
			&=	(-1)^{q-1}∂^U_g∂^U_f(ν⊗u🔑h⊗ω^♭)+(-1)^q∂^U_g∂^V_f(ν⊗u🔑h⊗ω^♭)\\
			&=	(-1)^q∂^U_f∂^U_g(ν⊗u🔑h⊗ω^♭)+(-1)^q∂^U_g∂^V_f(ν⊗u🔑h⊗ω^♭).
		\end{align*}
		Both terms in the last line
		are in the right-hand side of \cref{con:nog}.
		And we are done proving \cref{con:nofg}.
	\end{proof}
	
	In conclusion, this subsection proves that in order to
	repair two failing nodes at once, each of $d$ helper nodes sends out
	\[β_2=∑_{p+q=s-2}(d-k)^p\bi{k-1}q+∑_{p+q=s-2}(d-k)^p\bi{k-2}q
		\eqlabel{for:beta2}\]
	symbols.
	This quantity is the $x^s$-\coef/ of 
	\[÷{x^2(1+x)^{k-1}}{1-(d-k)x}+÷{x^2(1+x)^{k-2}}{1-(d-k)x}
		=x·÷{(1+x)^k-(1+x)^{k-2}}{1-(d-k)x}\eqlabel{fun:beta2}.\]

\subsection{Repair three failures}

	This subsection is similar to the previous.
	We bridge the gap between repairing two nodes
	and repairing arbitrary many nodes.
	
	Clearly the help messages are going to be
	the restrictions to coboundaries of the form
	$∂^U_e(\Vp⊗u🔑h⊗\Wq)$, $∂^U_f(\Vp⊗u🔑h⊗\Wq)$, and $∂^U_g(\Vp⊗u🔑h⊗\Wq)$,
	where $e,f,g$ are failing indices.
	The one and only problem is to bound the dimension from above.
	We claim that
	\[*∑∂^U_e(\Vp⊗u🔑h⊗\Wq)+∂^U_f(\Vp⊗u🔑h⊗\Wq)+∂^U_g(\Vp⊗u🔑h⊗\Wq)\]*
	is contained in
	\[*∑∂^U_e(\Vp⊗u🔑h⊗\Wq^⟂_e)+∂^U_f(\Vp⊗u🔑h⊗\Wq^⟂_{ef})
		+∂^U_g(\Vp⊗u🔑h⊗\Wq^⟂_{efg}).\]*
	Both sums are over $p+q=s-2$.
	To better illustrate the argument, we write $Ω$ to mean $\Vp⊗u🔑h⊗\Wq$.
	Similarly, $Ω_e$, $Ω_{ef}$, and $Ω_{efg}$ mean $\Vp⊗u🔑h⊗\Wq^⟂_e$,
	$\Vp⊗u🔑h⊗\Wq^⟂_{ef}$ and $\Vp⊗u🔑h⊗\Wq^⟂_{efg}$, respectively.
	Then what we want to prove turns into
	\[∑_{p+q=s-2}∂^U_e(Ω)+∂^U_f(Ω)+∂^U_g(Ω)⊆∑_{p+q=s-2}
		∂^U_e(Ω_e)+∂^U_f(Ω_{ef})+∂^U_g(Ω_{efg}).\eqlabel{con:noefg}\]
	Cf.\ \cref{con:nof,con:nofg}.
	
	Now the pattern is more recognizable.
	The mission is to remove $w🔑e$ from all three,
	$w🔑f$ from the last two, and $w🔑g$ from the last one.
	To achieve that, we claim the following six auxiliary containments:
	\begin{align*}
		∑∂^U_e(Ω) &⊆ ∑∂^U_e(Ω_e),\steplabel{con:start}\\
		∑∂^U_f(Ω_f) &⊆ ∑∂^U_e(Ω)+∂^U_f(Ω_{ef}),\\
		∑∂^U_f(Ω) &⊆ ∑∂^U_f(Ω_f),\\
		∑∂^U_g(Ω_{fg}) &⊆ ∑∂^U_e(Ω_f)+∂^U_g(Ω_{efg}),\\
		∑∂^U_g(Ω_g) &⊆ ∑∂^U_f(Ω)+∂^U_g(Ω_{fg}),\rlap{\qquad and}\\
		∑∂^U_g(Ω) &⊆ ∑∂^U_g(Ω_g).\steplabel{con:end}
	\end{align*}
	All sums are over $p+q=s-2$.
	The first, the third, and the last are equivalent to \cref{con:nof}.
	The second and the fifth are equivalent to \cref{con:nog}.
	The fourth looks exceptional, but is not---%
	we intend to remove $w🔑e$ from a $\Wq$-segment that lacks $w🔑f$.
	Whether or not $W$ contains $w🔑f$ in the first place
	is remotely relevant to our intention.
	In other words, the fourth containment is equivalent to \cref{con:nog}
	with $f$ replaced by $e$ and $W$ replaced by $W^⟂_f$.
	Colloquially, the repairment of two failures covers the most difficult case.
	
	All six auxiliary containments verified,
	we now elucidate how they together imply \cref{con:noefg}.
	Define a sudoku game as follows.
	The involving subspaces are placed in a grid as seen in \cref{fig:sudoku}.
	The highlighted subspaces are what is \emph{known}.
	At every step, we \emph{learn} a subspace if its prerequisites are known.
	Learning is governed by two rules.
	Rule-\ref{con:nof}: if a subspace has no upper neighbor,
	its right neighbor is the prerequisite.
	Rule-\ref{con:nog}: for the remaining subspaces, its right neighbor and
	the farthest subspace above it are the prerequisites.
	See \cref{fig:learn}.
	\Cref{con:noefg} is equivalent to whether
	the sudoku game ends at all subspaces highlighted.
	And it clearly does.
	
	\tikzset{know/.list={1-1,1-2,1-3}}
	\begin{figure}
		$$\tikz[sudoku]{
			\foreach\y in{1,...,3}{
				\edef\xmax{\iflabor\numexpr\y+1\else1\fi}
				\foreach\x in{1,...,\xmax}{
					\node at(-\x,-\y)[subspace cell](\x-\y){$∂^U_\y(Ω_{
						\ifnum\x>\y\else\foreach\zz in{\x,...,\y}{\zz}\fi
					})$};
				}
			}
		}$$\caption{
			The sudoku game board for $c=3$.
			For simplicity, $e,f,g$ are replaced by $1,2,3$.
			The goal is to learn the entire board.
		}\label{fig:sudoku}
	\end{figure}
	
	\begin{figure}
		\tikzset{learn=2-1}
		$$\tikz[sudoku,shorten both=1]{
			\foreach\y in{1,...,3}{
				\edef\xmax{\iflabor\numexpr\y+1\else1\fi}
				\foreach\x in{1,...,\xmax}{
					\node at(-\x,-\y)[subspace cell](\x-\y){$∂^U_\y(Ω_{
						\ifnum\x>\y\else\foreach\zz in{\x,...,\y}{\zz}\fi
					})$};
				}
			}
			\iflabor
			\draw[->,overlay](1-1)to[bend right=60,looseness=1.5](2-1);
			\fi
		}\qquad
		\tikzset{know/.list={2-1},learn=2-3}
		\tikz[sudoku,shorten both=1]{
			\iflabor
			\foreach\y in{1,...,3}{
				\edef\xmax{\iflabor\numexpr\y+1\else1\fi}
				\foreach\x in{1,...,\xmax}{
					\node at(-\x,-\y)[subspace cell](\x-\y){$∂^U_\y(Ω_{
						\ifnum\x>\y\else\foreach\zz in{\x,...,\y}{\zz}\fi
					})$};
				}
			}
			\draw[->](2-1)to[bend right](2-3);
			\draw[->,overlay](1-3)to[bend left=60,looseness=1.5](2-3);
			\fi
		}$$\caption{
			To the left: rule-\ref{con:nof}.
			To the right: rule-\ref{con:nog}.
			Subspaces are expected to be learned in this order:
			$∂^U_1(Ω)$, $∂^U_2(Ω_2)$, $∂^U_2(Ω)$, $∂^U_3(Ω_{23})$,
			$∂^U_3(Ω_3)$, and finally $∂^U_3(Ω)$.
			This is the exact order
			\crefrange{con:start}{con:end} are given in.
		}\label{fig:learn}
	\end{figure}
	
	To sum up, this subsection proves that in order to
	repair three failing nodes at once, each of the $d$ helper nodes sends out
	\[β_3=∑_{p+q=s-2}(d-k)^p\bi{k-1}q+(d-k)^p\bi{k-2}q+(d-k)^p\bi{k-3}q
		\eqlabel{for:beta3}\]
	symbols.
	Cf.\ \cref{for:beta,for:beta2}.
	This quantity is the $x^s$-\coef/ of 
	\[÷{x^2(1+x)^{k-1}}{1-(d-k)x}+÷{x^2(1+x)^{k-2}}{1-(d-k)x}
		+÷{x^2(1+x)^{k-3}}{1-(d-k)x}=x·÷{(1+x)^k-(1+x)^{k-3}}{1-(d-k)x}.
		\eqlabel{fun:beta3}\]
	Cf.\ \cref{fun:beta,fun:beta2}.

\subsection{Repair many failures}

	Let there be $c$ failing nodes.
	If $c≥k$, then it is sufficient and necessary
	to send out the entire file to the agent.
	We know, from the axioms,
	that any $k$ out of $d$ helpers sending out $α$ symbols is enough.
	How to further reduce this number or balance the load among the $d$ helpers
	could be an interesting subject but is not addressed here.
	
	Consider the $c<k$ cases.
	Owing to the innate symmetry, assume the first $c$ nodes fail.
	Let $W^⟂≔W$;
	let $W^⟂_{ℐj}$ for $j∈[c]$ be such that
	$W^⟂_ℐ≅W^⟂_{ℐj}⊕\spa⟨w🔑j⟩$ for all $ℐ⊆[c]$ and all $j∈[c]、ℐ$.
	Informally, $W^⟂_ℐ$ is the space $W$ without $w🔑i$ for $i∈ℐ$.
	It could be rigorously implemented as
	$W^⟂_ℐ≔\spa⟨w🔑j:j∈[k]、ℐ⟩$.
	Let $Ω_ℐ$ be $\Vp⊗u🔑h⊗\Wq^⟂_ℐ$. 
	
	The containment we want is
	\[∑_{p+q=s-2}∑_{j=1}^c∂^U_j(Ω)⊆∑_{p+q=s-2}∑_{j=1}^c∂^U_j(Ω_{12\dotsm j}).
		\eqlabel{con:no12j}\]
	Cf.\ \cref{con:nof,con:nofg,con:noefg}. 
	The containment can be proven by a similar sudoku game.
	See \cref{fig:sudoku,fig:learn,fig:sudokularge}.
	The inequality we want is
	\[\dim（†right-hand side of \eqref{con:no12j}†）
		≤∑_{p+q=s-2}∑_{j=1}^c(d-k)^p\bi{k-j}q.\eqlabel{ine:betac}\]
	Cf.\ \cref{ine:lessk,ine:2lessk}.
	The right-hand side of \cref{ine:betac} is \cref{for:betamore}.
	(Apply Pascal's rule iteratively.)
	Cf.\ \cref{for:beta,for:beta2,for:beta3}.
	It is the $x^s$-\coef/ of
	\[*∑_{j=1}^c÷{x^2(1+x)^{k-j}}{1-(d-k)x}
		=x·÷{(1+x)^k-(1+x)^{k-c}}{1-(d-k)x},\]*
	which coincides with \cref{fun:betamore}.
	(Sum the geometric series.)
	Cf.\ \cref{fun:beta,fun:beta2,fun:beta3}.
	This finishes the proof of \cref{pro:more}.
	
	\tikzset{know/.list={1-4,1-5}}
	\begin{figure}
		$$\tikz[sudoku,xscale=1.2,yscale=1/1.2]{
			\foreach\y in{1,...,5}{
				\edef\xmax{\iflabor\numexpr\y+1\else1\fi}
				\foreach\x in{1,...,\xmax}{
					\node at(-\x,-\y)[subspace cell](\x-\y){$∂^U_\y(Ω_{
						\ifnum\x>\y\else\foreach\zz in{\x,...,\y}{\zz}\fi
					})$};
				}
			}
		}$$\caption{
			The sudoku game board for $c=5$.
			The goal is to learn the entire board,
			presumably in this order: from right to left, from top to bottom.
		}\label{fig:sudokularge}
	\end{figure}

讀
\section
								   Discussion   								

	While moulin code is ``numerically equivalent'' to cascade code,
	we argue that our codes possess (much) more structure.
	For one, our file format consists of a chain of spaces
	quotiented by the ``joints'';
	in contrast, cascade's file format is tree-like in nature.
	For two, the downloading scheme is an induction;
	the induction is driven by a single parameter $p$ (equivalently $q$)
	and uses merely direct sum's universal property.
	Three, the repairing rule is an explicit variable assignment
	(\cref{equ:repair}) instead of a segment-by-segment procedure.
	(The proof of the validity needs induction, but the algorithm does not.)
	Four, we analyze how moulin code withstands simultaneous node failures.
	This has been done only for $k=d$ determinant code \cite{EM19d}.
	For $k<d$ cascade code, however, it is easier claimed than proven.
	
	Equipped with structures, moulin can be used to
	backtrack which step of the outer bounds is not tight.
	This potentially helps tighten outer bounds
	and figure out the optimal $α/M$-versus-$β/M$ trade-off.

\appendix

讀
\section
					    Connection to Other Constructions   					

\label{sec:connection}

	We arrange eight recent works into \cref{tab:paradigm}.
	They are: layered code \cite{TSAVK15},
	improved layered code \cite{SSK15},
	concatenated layered code
	(using Johnson graph code as outer code) \cite{DL19},
	determinant code \cite{EM16d},
	cascade code \cite{EM19c},
	a multilinear algebra description of layered code \cite{LD17},
	an unpublished code by us,
	and moulin code discussed in this work.
	
	Codes in the combinatorics column are limited to the $n=d+1$ case,
	while the others apply to all $n≥d+1$ cases (with a price of larger fields).
	When $d$ is set to be $k$, the first column collapses into layered;
	the others degenerate to determinant.
	The former (layered) consists of \emph{layers};
	each layer is an $[s,s-1]$-\MDS/ code whose $s$ symbols
	are distributed across some $s$ nodes.
	The latter (determinant) consists of \emph{virtual} layers
	where the symbols are linearly combined and distributed to all nodes.
	
	Codes in the first row achieve the optimal trade-off when $k=d$.
	They (layered and determinant) are the building blocks
	of the other codes in the corresponding columns.
	Improved layered provides a rule of thumb to chain/nest layers
	such that overprotected fragments help insecure fragments.
	Codes in the third row specify a way to chain/nest
	their ``first-row primitives'' such that
	the overall $(α,β,M)$ reaches the best known possibilities.
	
	\pgfplotstableread[header=false]{
		~						combinatorics	Elyasi--Mohajer	multilinear			
		{$k=d$}					layered			determinant		\cite{LD17}			
		{$k≤d$~(suboptimal)}	improved		--				--					
		{$k≤d$~(conj.\ best)}	concatenated	cascade			unpublished/moulin	
	}\pgfParadigm
	\begin{table}
		\caption{
			Eight general regimes of ERRCs.
			By general we mean the targeted parameters
			are beyond the \MBR/ and \MSR/ points.
		}\label{tab:paradigm}
		\def\arraystretch{1.44}
		\centering\pgfplotstabletypeset[
			every head row/.style=output empty row,
			every row no 0/.style={before row=\toprule,after row=\midrule},
			every last row/.style={after row=\bottomrule},
			string type,
		]\pgfParadigm
	\end{table}
	
	This is how improved layered chains layers together:
	For each insecure (formerly \emph{under-accessed}) layer,
	count the degree of insecurity (corank/equivocation/deficit in protection),
	which is the number of absent symbols in this layer.
	For each absence of symbol, invoke an $s=1$ layer to back it up;
	an $s=1$ layer is always overprotected (formerly \emph{fully-accessed}).
	
	This is how concatenated layered chains layers together:
	For each insecure layer of size $s=t$,
	count the insecurity and let it be $r$.
	For each absence of symbol, back it up by a layer of size $s≤t-1-r$.
	Some flexibility lives within the choices of layer sizes.
	
	This is how cascade chains virtual layers together:
	For each insecure virtual layer of size $s=t$,
	let $r$ be the insecurity.
	Highlight the lexicographically last $r$ symbols in this virtual layer.
	For each of those symbols, prepare an $s≤t-1-r$ layer and
	constrain that the sum of the smaller layer is the highlighted symbol.
	
	This is how the unpublished chains virtual layers together:
	All virtual layers of size $s$ and insecurity $r$ span a vector space.
	We assign a subspace that is always
	linearly independent from the secure part.
	Call this subspace $𝘔(r)$,
	and call the protecting virtual layer $𝘋(t-1-r)$.
	Then the tensor product $𝘔(r)⊗𝘋(t-1-r)$ is the backup.
	In this case, one may say that
	there are (effectively) $\dim(𝘔(r))$ copies of $𝘋(t-1-r)$ in charge.
	That being said, it is hard to choose a basis of $𝘔(r)$
	in order to identify the copies;
	and we did not pursue such choice.
	The symmetry comes from the fact that $𝘔(r)$
	does not depend on any basis or star vector;
	in fact, $𝘔(r)$ was made the null space of some \WM/ map.
	
	All codes above chain (virtual) layers of parameters $(n,d,d)$.
	Moulin in this paper chains virtual layers of parameters $(n,k,k)$.
	Each \CM/ map one sees in \cref{dia:play4} is the parametrization
	space $\Vp$ tensored with the virtual layer $W⊗\Wq→\W{q+1}$.
	The parametrization space creates $\dim(\Vp)$ copies of
	the virtual layer and indexes them by tensors in $\Vp$.
	
	Another angle that reveals the difference between moulin and the others
	is how ogfs are obtained:
	In \cite[Fig.~4]{EM19c}, the numbers of determinant codes
	of virtual layer sizes $5,4,3,2,1$ are $1,0,3,2,9$, respectively.
	The latter sequence is a truncation of an indefinite one
	$1,0,3,2,9,12,31,54,\allowbreak\dotsc$
	generated by the ogf \cite{A053088} 
	\[*÷1{(1-2x)(1+x)^2}=÷1{(1-(d-k)x)(1+x)^{d-k}}.\]*
	For virtual layers of sizes $5,4,3,2,1$, they cost each node
	$\binom d4,\binom d3,\binom d2,\binom d1,\binom d0$ symbols, respectively.
	This is generated by
	\[*x(1+x)^d.\]*
	Therefore, the node capacity $α$ is the $x^5$-\coef/ of the product
	\[*÷1{(1-(d-k)x)(1+x)^{d-k}}·x(1+x)^d=÷{x(1+x)^k}{1-(d-k)x}.\]*
	Compare this to how we derived \cref{fun:alpha} in \cref{sec:node4}:
	\[*÷1{1-(d-k)x}·x(1+x)^k=÷{x(1+x)^k}{1-(d-k)x}.\]*

讀
\section
									 Gallery									

\label{sec:gallery}

	This appendix displays some siblings of \cref{dia:play4}---%
	the specialization of \cref{dia:play} to the first few sizes $s$.
	For each of them, the size is the number at the top left cell.
	For instance, $\V4$ at the top left cell means $s=4$.

\iflabor
	$s=1$.
	This is not useful as a standalone ERRC
	because everything vanishes by parity checks.
	Nonetheless, it is always a part of higher size moulin.
	And it protects.
	$$\CD[skew 3=]{
		\W1\ARCM & \V1          \\
		         & U\ARDS       \\
		         & \quad W\quad 
	}$$
	
	$s=2$.
	This is moulin at the \MBR/ point.
	This coincides with the product matrix construction at the \MBR/ point
	\cite[Section~IV]{RSK11}.
	$$\CD[skew 3=]{
		\W2\ARCM & \VW11      & \V2        \\
		         & U⊗\W1\ARDS & \V1⊗U\ARDS \\
		         & W⊗\W1       & \V1⊗W      
	}$$
	
	$s=3$.
	$$\CD[skew 3=]{
		\W3\ARCM & \VW12      & \VW21  & \V3        \\
		         & U⊗\W2\ARDS & \VUW11 & \V2⊗U\ARDS \\
		         & W⊗\W2      & \VWW11 & \V2⊗W      
	}$$
	
	$s=4$.
	This is the example \cref{dia:play4} we gave in the main text.
	$$\CD[skew 3=]{
		\W4\ARCM & \VW13      & \VW22  & \VW31  & \V4        \\
		         & U⊗\W3\ARDS & \VUW12 & \VUW21 & \V3⊗U\ARDS \\
		         & W⊗\W3      & \VWW12 & \VWW21 & \V3⊗W      
	}$$
	
	$s=5$.
	When $(k,d)=(4,6)$, this ERRC has the same parameters
	as the illustrative example in \cite[section~VII]{EM19c}.
	When $(n,k,d)=(8,4,7)$, this ERRC has the same parameters
	as in \cite[section~7]{DL19}.
	$$縮\CD[skew 3=-9]{
		\W5\ARCM & \VW14      & \VW23  & \VW32  & \VW41  & \V5        \\
		         & U⊗\W4\ARDS & \VUW13 & \VUW22 & \VUW31 & \V4⊗U\ARDS \\
		         & W⊗\W4      & \VWW13 & \VWW22 & \VWW31 & \V4⊗W      
	}縮$$
\fi

	Moulin at the \MSR/ point is \emph{not} shown here
	because it depends on $k$---it is the $s=k+1$ one.
	Note that $\W{k+1}$ and higher wedge powers vanish.
	So all $s≥k+1$ construction will simply be the same.

\iflabor
	The following diagram shows in what direction
	\CM/ $∇$ and coboundary operators map.
	South west arrows are $∂^W_w$ (for any $w∈W$).
	South east arrows are $∂^V_v$ (for any $v∈V$).
	Lightning arrows are compositions of
	$U↠V$ followed by $∇$ and finally $W↪U$.
	One can visualize \cref{lem:homotopy} in this diagram.
	$$縮
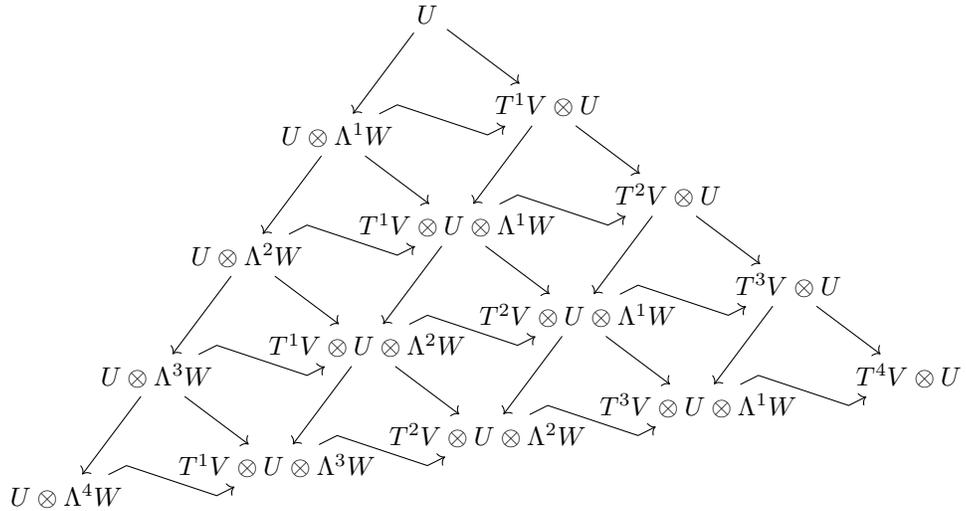
\tikz{
		\tikzset{x={(-3cm,-4cm)},y={(4cm,-3cm)},scale=.4}
		\foreach\p in{0,...,5}{
			\foreach\q in{0,...,\numexpr5-\p}{
				\ifnum\numexpr\p+\q<5
					\edef\prefix{\ifnum\p>0\V\p⊗\fi}
					\edef\suffix{\ifnum\q>0⊗\W\q\fi}
					\draw(\q,\p)node(\q-\p){$\prefix U\suffix$};
					\ifnum\p>0
						\PMT\pp{\p-1}
						\draw[->](\q-\pp)--(\q-\p);
					\fi
					\ifnum\q>0
						\PMT\qq{\q-1}
						\draw[->](\qq-\p)--(\q-\p);
					\fi
				\fi
				\ifnum\p>0\ifnum\q>0
					\PMT\pp{\p-1}\PMT\qq{\q-1}
					\draw[->](\q-\pp)--(\q-.4,\p-.8)--(\q-.6,\p-.2)--(\qq-\p);
				\fi\fi
			}
		}
	}縮$$
	The diagram goes on endlessly but we stop at $s=5$.
\fi

\hbadness9999
\makeatletter
\g@addto@macro\sloppy{\advance\baselineskip0ptplus1ptminus1pt}
\let\($	\let\)$
\bibliographystyle{alphaurl}
\bibliography{MoulinAlgebra-3}

\end{document}